\documentclass[10pt,twocolumn,letter]{IEEEtran}

\usepackage[latin9]{inputenc}
\usepackage{color}
\usepackage{amsmath}
\usepackage{amsthm}
\usepackage{amssymb}
\usepackage{graphicx}
\usepackage[T1]{fontenc}
\usepackage[active]{srcltx}
\usepackage{float}
\usepackage{setspace}
\usepackage{enumitem}
\usepackage[font=small]{caption}
\usepackage[font={footnotesize}]{subcaption}
\usepackage{epsfig}
\usepackage[noadjust]{cite}
\usepackage{verbatim}
\usepackage{balance}
\usepackage{stackrel}
\usepackage{epstopdf}
\usepackage{mathrsfs}
\usepackage{dsfont}
\usepackage{algorithm}
\usepackage{algpseudocode}

\DeclareMathAlphabet{\baz}{OML}{cmm}{b}{i}
\def\x{{\mathbf x}}

\def\bA{\mbox{\boldmath $A$}}
\def\bB{\mbox{\boldmath $B$}}

\def\bI{\mathbf{I}}
\def\bJ{\mathbf{J}}

\def\bW{\mathbf{W}}

\def\b0{\mbox{\boldmath $0$}}
\def\ba{\mathbf{a}}
\def\bb{\mathbf{b}}

\def\br{\mathbf r}

\def\bw{\mbox{\boldmath $w$}}
\def\bx{\mathbf{x}}

\def\by{\mathbf y}
\def\bz{\mathbf{z}}

\def\x{{\mathbf x}}

\newtheorem{theorem}{Theorem}
\newtheorem{proposition}[theorem]{Proposition}
\newtheorem{lemma}[theorem]{Lemma}

\newtheorem{definition}[theorem]{Definition}

\floatstyle{ruled}
\newfloat{algorithm}{tbp}{loa}
\providecommand{\algorithmname}{Algorithm}
\floatname{algorithm}{\protect\algorithmname}

\let\oldproofname=\proofname
\renewcommand{\proofname}{\rm\bf{\oldproofname}}

\title{\vspace{-0.5cm} NEXT: In-Network Nonconvex Optimization}

\author{Paolo~Di Lorenzo,~\IEEEmembership{Member,~IEEE}, and Gesualdo Scutari,~\IEEEmembership{Senior Member,~IEEE}\vspace{-0.7cm}
\thanks{Di Lorenzo is with the Dept. of Engineering, University of Perugia, Via G. Duranti 93, 06125, Perugia, Italy; Email: \texttt{paolo.dilorenzo@unipg.it}.\newline \indent Scutari is with the Dept. of Industrial Engineering and Cyber Center (Discovery Park), Purdue University, West Lafayette, IN 47907, USA; E-mail: \texttt{gscutari@purdue.edu}.  The work of Scutari was supported by the USA NSF under Grants CIF 1564044 and  CAREER Award No. 1555850, and ONR Grant  N00014-16-1-2244.}}

\begin{document}

\maketitle

\begin{abstract}
We study  \emph{nonconvex distributed} optimization in multi-agent networks with time-varying (nonsymmetric) connectivity. We introduce  the first algorithmic framework for the distributed minimization of  the sum of a  smooth (possibly \emph{nonconvex} and nonseparable) function--the agents' sum-utility--plus a convex (possibly nonsmooth and nonseparable) regularizer. The latter is usually employed to enforce some structure in the solution, typically sparsity. The proposed method hinges on successive convex approximation techniques while leveraging \emph{dynamic} consensus as a mechanism to distribute the computation among the agents: each agent first  solves (possibly inexactly) a local convex approximation of the nonconvex original problem, and then performs  local averaging operations. Asymptotic convergence to (stationary) solutions of the nonconvex problem is established. Our algorithmic framework is then customized  to a variety of convex and nonconvex problems in  several fields, including   signal processing, communications, networking, and machine learning.  Numerical results  show that the new method compares favorably to existing distributed algorithms on both convex and nonconvex  problems.\vspace{-0.1cm}
\end{abstract}

\begin{keywords}
Consensus,  distributed optimization, non-convex optimization, successive convex approximation, time-varying directed graphs.\vspace{-0.5cm}
\end{keywords}

\section{Introduction}\label{sec:intro}

Recent years have witnessed a surge of interest in distributed optimization methods for multi-agent systems. Many such problems can be formulated as the cooperative minimization of the agents' sum-utility $F$ plus  a regularizer $G$:
\begin{align}\label{Problem}
& \underset{\bx}{\min} \;\; U(\bx)\triangleq F(\bx)+G(\bx)\\
& \text{\,s.\,t.} \;\; \bx\in \mathcal{K},\nonumber
\end{align}
where\vspace{-0.3cm}
\begin{equation} F(\bx)\triangleq \sum_{i=1}^{I}f_i(\bx),
\end{equation}
with each $f_i: \mathbb{R}^m \rightarrow  \mathbb{R}$ being the smooth (possibly \emph{nonconvex}, nonseparable) cost function of agent $i\in \{1,\ldots, I\}$;  $G$ is a  convex (possibly nonsmooth,  nonseparable) function; and $\mathcal{K}\subseteq \mathbb{R}^m$ is closed and convex. Usually the nonsmooth term
is used to promote some extra structure in the solution; for instance,   $G(\bx)=c\|\bx\|_1$ or $G(\bx)=c \sum_{i=1}^N\|\bx_i\|_2$ are widely used to impose (group)  sparsity of the  solution. \\
Network-structured optimization problems in the form (\ref{Problem}) are found widely in several engineering areas, including sensor networks information processing (e.g., parameter estimation, detection, and localization), communication networks (e.g., resource allocation in peer-to-peer/multi-cellular systems), multi-agent control and coordination (e.g., distributed learning, regression, and flock control), and distributed machine learning (e.g., LASSO, logistic regression, dictionary learning, matrix completion, tensor factorization), just to name a few. Common to these problems is the necessity of performing  a completely decentralized  computation/optimization.
For instance, when data are collected/stored in a distributed network (e.g., in clouds), sharing local information with a central processor   is either unfeasible or not economical/efficient, owing to the large size of the network and volume of data, time-varying network topology, energy constraints, and/or privacy issues.  Performing the optimization in a centralized fashion may raise   robustness concerns as well, since  the central processor represents an isolate point of failure. Motivated by these observations,  this paper aims to develop a solution method with provable convergence for the general class of \emph{nonconvex} problems (\ref{Problem}), in the following distributed setting:  i) the network of agents is modeled as a time-varying directed graph; ii) agents know their local functions $f_i$ only, the common regularized $G$, and the feasible set $\mathcal K$; and iii) only inter-node (intermittent) communications between single-hop neighbors are possible. Hereafter, we will call \emph{distributed} an algorithm implementable in the above setting. To the date, the design of  such an algorithm for Problem (\ref{Problem}) remains a    challenging and open problem, as documented next.

\noindent\textbf{Related works:} Distributed solution methods  for \emph{convex} instances of Problem (\ref{Problem}) have been widely studied in the literature; they are usually either primal (sub)gradient-based methods or primal-dual schemes.
Algorithms belonging to the former class  include: i) consensus-based (sub)gradient schemes  \cite{Tsitsiklis-Bertsekas-Athans,Nedic-Ozdaglar, Nedic-Ozdaglar-Parillo,Shi-Ling-Wu-Yin} along with  their accelerated and asynchronous versions \cite{Chen-Ozdaglar,Jakovetic-Xavier-Moura,Srivastava-Nedic};
ii) the (sub)gradient push-method \cite{Nedic-Olshevsky,Tsianos-Rabbat}; iii)  the dual-average method \cite{Tsianos-Lawlor-Rabbat2,Duchi-Agarwal-Wainwright}; and iv) distributed second-order-based schemes \cite{Zanella,Mokhtari-Ling-Ribeiro}. Algorithms for adaptation and learning tasks based on in-network diffusion techniques were proposed in \cite{Cattivelli-Sayed,DiLo-Sayed,DiLo14,Chen-Sayed,Sayed}.
Although there are substantial differences between them, these methods can be generically abstracted as combinations of local  descent steps followed by variable exchanges and averaging of information among neighbors. The second class of  distributed algorithms  is that of  dual-based techniques. Among them, we mention here only the renowned Alternating Direction Method of Multipliers (ADMM); see \cite{Boyd-Parikh-Chu-Peleato-Eckstein} for a recent survey. Distributed ADMM algorithms tailored for specific machine learning \emph{convex} problems and parameter estimation in sensor networks were proposed in \cite{Forero-Cano-Giannakis, Schizas-Giann-Roum-Rib}; \cite{Mota-Xavier,Shi-Ling}, \cite{ChangHW15}, studied the  convergence of synchronous distributed ADMM over undirected connected time-invariant graphs;  some asynchronous instances  have been recently analyzed in \cite{Wei-Ozdaglar13}, \cite{iutzeler2013cdc}.

All the above  prior art focuses   \emph{only} on \emph{convex} problems; algorithms developed therein  along with their convergence analysis    are not applicable to nonconvex problems in the form   (\ref{Problem}). Parallel  and partially decentralized  solution methods  for  some families of  nonconvex problems have been recently proposed in   \cite{Scutari-Facchinei-Song-Palomar-Pang,Scutari-Facchinei-Sagratella,Dan-Facch-Kung-Scut,Patterson-Eldar-Keidar,Ravazzi-Fosson-Magli,Li-Scaglione}. However,  these methods  are not applicable to  the general formulation     (\ref{Problem}) and/or in the distributed setting i)-iii) discussed above. For instance,   some of them   require    the  knowledge   of the whole $F$  (or its derivative) from all  the agents; others call for the presence of a fusion center collecting at each iteration data from all the agents; some others  are implementable only on specific network topologies, such as fully connected (undirected) graphs (i.e., agent must be able to exchange information with \emph{all}  the others). We are aware of only a few works dealing with distributed algorithms for (special cases of) Problem (\ref{Problem}), namely: \cite{Zhu-Martinez2} and \cite{Bianchi-Jakubowicz}.
In \cite{Zhu-Martinez2}, a consensus-based distributed dual-subgradient algorithm was studied. However,  i)  it calls for the solution of possibly difficult nonconvex (smaller) subproblems; ii) it does not find (stationary) solutions of the original problem but  those of an auxiliary problem; stationary points of this reformulation are  not necessarily stationary for the original problem; and iii) convergence of primal variables is guaranteed under some restrictive assumptions that are not easy to be checked a-priori.
In \cite{Bianchi-Jakubowicz}, the authors studied convergence of a distributed stochastic projection algorithm, involving random gossip between agents and diminishing step-size.
However, the scheme as well as   its convergence analysis are not applicable to Problem (\ref{Problem}) when $G\neq 0$.
Moreover, it is a gradient-based algorithm, using thus only   first order
information of $f_i$; recently it was shown in \cite{Scutari-Facchinei-Song-Palomar-Pang,Scutari-Facchinei-Sagratella,Dan-Facch-Kung-Scut}  that exploiting the structure of the nonconvex functions by replacing their linearization with a ``better'' approximant can enhance practical convergence speed; a fact that we would like to exploit in our design.

\noindent \textbf{Contribution:} This paper introduces the first \emph{distributed} (best-response-based) algorithmic framework with provable convergence  for the  \emph{nonconvex} multi-agent optimization in the general form (\ref{Problem}).
The crux of the framework is a novel \emph{convexification-decomposition} technique that hinges on our recent (primal) Successive Convex Approximation (SCA) method \cite{Scutari-Facchinei-Song-Palomar-Pang,Scutari-Facchinei-Sagratella}, while leveraging  \emph{dynamic} consensus (see, e.g., \cite{Zhu-Martinez}) as a mechanism to distribute the computation as well as propagate the needed information over the network; we will term it as in-Network succEssive conveX approximaTion algorithm (NEXT). More specifically, NEXT is based on the (possibly inexact) solution from each agent of a  sequence of strongly convex, decoupled, optimization subproblems, followed by a consensus-based update. In each subproblem,  the nonconvex sum-utility $F$ is replaced by  a (strongly) convex surrogate that can be \emph{locally} computed by the agent, independently from the others. Then, two steps of consensus are performed to force respectively an agreement among users' local solutions and update some parameters in the surrogate functions. While leveraging  consensus/diffusion methods   to align   local users'  estimates has been widely explored in the literature, the use of dynamic consensus  to update the objective functions of users' subproblems is a novel idea, introduced for the first time in this paper, which makes the proposed scheme convergent even in the case of nonconvex $F$'s.
Some remarkable  novel features of NEXT are: i) it is very flexible in the choice of the approximation of $F$, which need not be necessarily its first or second order approximation (like in all current consensus-based  schemes); of course it includes, among others, updates based
on gradient- or Newton-type approximations; ii) it allows for inexact solutions of the subproblems;  and iii) it deals with nonconvex and nonsmooth objectives in the form $F+G$.
The proposed framework encompasses a gamut of novel algorithms, all converging under the same conditions. This offers a lot of
flexibility to tailor the method to specific problem structures and to control the signaling/communication overhead. We illustrate several potential applications in different areas,  such as distributed signal processing, communications, and networking. Numerical results show that  our schemes outperform current ones in terms of practical convergence while reaching the same (stationary) solutions. Quite remarkably, this has been observed also for convex problems, which was
not obvious at all, because existing algorithms heavily rely
on the convexity of the problem, whereas our framework has
been designed to handle (also) nonconvex problems.  As a final remark, we underline that, at more methodological level, the combination of SCA techniques \cite{Scutari-Facchinei-Song-Palomar-Pang,Scutari-Facchinei-Sagratella}, and dynamic consensus \cite{Zhu-Martinez} and, in particular, the need to conciliate the use of surrogate functions with local updates,  led to the development of a new type of convergence analysis, which does not rely on convexity properties of the
utility functions,  and is also of interest per se and could bring to further developments.

The paper is organized as follows.  Section II contains the main theoretical results of the paper: we start with an informal, constructive description of the algorithm (cf. Sec. II.A), and then introduce formally the  framework along with its convergence properties (cf. Sec. II.B). Section III generalizes NEXT to more general settings.  Section IV customizes NEXT to a variety of   practical problems, arising from applications in signal processing, machine learning, and networking; it also compares numerically our schemes with prior algorithms.   Finally, Section V draws some conclusions.\vspace{-0.35cm}

\section{A New In-network Optimization Technique}

Consider a network composed of $I$ autonomous agents aiming to cooperatively and distributively solve Problem (\ref{Problem}).

\smallskip
\noindent \textbf{Assumption A [On Problem (\ref{Problem})]:}
\begin{description}[topsep=-2.0pt,itemsep=-2.0pt]
\item[(A1)]  The set $\mathcal{K}$ is (nonempty) closed and convex;\smallskip
\item[(A2)] Each $f_i$ is $C^1$ (possibly nonconvex) on an open set containing $\mathcal K$;\smallskip
\item[(A3)]  Each $\nabla f_i$ is Lipschitz continuous on $\mathcal{K}$, with constant $L_i$; let $L^{\max}\triangleq \max_i L_i;$\smallskip
\item[(A4)] {$\nabla F$ is bounded on $\mathcal{K}$: there exists a finite scalar $L_F>0$ such that  $\|\nabla F(\bx)\|\leq L_F$, for all  $\bx\in \mathcal K$;}\smallskip
\item[(A5)] {$G$ is a convex function (possibly nondifferentiable) with bounded subgradients on $\mathcal K$: there exists a finite scalar $L_G>0$ such that $\|\partial G(\bx)\|\leq L_G$, for any subgradient $\partial G(\bx)$ of $G$ at any $\bx\in \mathcal K$};\smallskip
\item[(A6)] {$U$ is coercive on $\mathcal K$, i.e., $\lim_{\bx\in \mathcal K, \|\bx\|\rightarrow \infty} U(\bx)=+\infty$.}\smallskip
\end{description}

Assumption A is standard and  satisfied by many practical problems. For instance, A3-A4
hold automatically if $\mathcal{K}$ is bounded, whereas A6 guarantees the existence of a solution. Note that $f_i$'s need not be convex; moreover, no knowledge of $L_F$ and $L_G$ is required. We also make the blanket assumption that each agent $i$   knows  only its  own cost function $f_i$ (but not $F$), the common $G$, and the feasible set $\mathcal K$.

\noindent {\textbf{On network topology:}  Time is slotted,  and at any time-slot $n$, the network is modeled as a digraph $\mathcal{G}[n]=(\mathcal{V,E}[n])$, where $\mathcal{V}=\{1,\ldots,I\}$ is the vertex  set (i.e., the set of agents), and $\mathcal{E}[n]$ is the set of (possibly) time-varying directed edges. The  in-neighborhood of agent $i$ at time $n$ (including node $i$) is defined as  $\mathcal{N}_i^{\rm in}[n]=\{j|(j,i)\in\mathcal{E}[n]\}\cup\{i\}$; it sets  the communication pattern  between single-hop neighbors:
 agents  $j\neq i$ in $\mathcal{N}_i^{\rm in}[n]$ can communicate with node $i$ at time $n$. Associated with each graph $\mathcal{G}[n]$, we introduce (possibly) time-varying weights $w_{ij}[n]$ matching  $\mathcal{G}[n]$:
\begin{align}\label{weights}
w_{ij}[n]=\left\{
             \begin{array}{ll}
               \theta_{ij}\in[\vartheta,1] & \hbox{if $j\in \mathcal{N}_i^{\rm in}[n]$;} \\
                0 & \hbox{otherwise,}
             \end{array}
           \right.
\end{align}
\noindent for some $\vartheta\in (0,1)$, and define the matrix $\bW[n]\triangleq (w_{ij}[n])_{i,j=1}^I$. These weights will be used later on in definition of the proposed algorithm.}

\noindent {{\bf  Assumption B (On the network topology/connectivity):} % and agents' knowledge/signaling:
\begin{description}[topsep=-2.0pt,itemsep=-2.0pt]
\item[(B1)]  The   sequence of graphs $\mathcal{G}[n]$ is B-strongly connected, i.e., there exists an integer $B > 0$ such that the graph $\mathcal{G}[k]=(\mathcal{V},\mathcal{E}_B[k])$, with
$\mathcal{E}_B[k]=\bigcup_{n=kB}^{(k+1)B-1}\mathcal{E}[n]$
 is strongly connected, for all $k\geq0$; \smallskip
\item[(B2)] Every weight matrix $\bW[n]$ in (\ref{weights})   satisfies
\begin{align}\label{Doubly_Stochastic}
\bW[n]\,\mathbf{1}=\mathbf{1} \quad \text{and}\quad \mathbf{1}^T \bW[n]=\mathbf{1}^T \quad \forall n.
\end{align}
\end{description}}

Assumption B1 allows strong connectivity to occur over a long time period and in arbitrary order.    Note also that $\bW[n]$ can be   time-varying and need not be symmetric.

Our goal is to develop an algorithm that converges to stationary solutions of Problem (\ref{Problem}) while being implementable  in the above  distributed setting (Assumptions A and B). \vspace{-0.2cm}
\begin{definition} A point $\x^*$  is a stationary solution of Problem (\ref{Problem}) if    {a subgradient $\partial G(\x^*)$ exists such
that
$(\nabla F(\x^*) $ $+ \partial G(\x^*) )^T(\mathbf{y}-\x^*) \geq 0$, for all $\mathbf{y}\in \mathcal K$.} \\Let $\mathcal S$ be the set of stationary solutions of (\ref{Problem}).\vspace{-0.2cm}
\end{definition}
To shed light on the core idea of the novel decomposition technique, we begin by introducing in Sec. \ref{sec:informal_description} an informal and constructive  description of our scheme. Sec. \ref{sec:NEXT_algo} will formally introduce  NEXT along with its convergence properties. The inexact version of the scheme is discussed in Sec. \ref{sec:extensions}.  \vspace{-0.2cm}

\vspace{-0.25cm}
\subsection{Development of NEXT: A constructive approach}\label{sec:informal_description}
Devising  distributed solution methods for Problem (\ref{Problem}) faces two main challenges, namely: the nonconvexity of $F$ and the lack of global information on $F$. To cope with these issues, we propose to combine SCA techniques  (Step 1 below)  with consensus mechanisms (Step 2), as described next.

\noindent \textbf{Step 1 (local SCA optimization):} Each agent $i$ maintains a local estimate $\bx_i$ of the optimization variable $\bx$ that is iteratively updated. {Solving directly Problem (\ref{Problem}) may be too costly (due to the nonconvexity of $F$) and is  not even feasible in a distributed setting (because of the lack of knowledge of the whole $F$). One may then prefer to approximate Problem (\ref{Problem}), in some suitable sense, in order
to permit each agent to compute \emph{locally} and \emph{efficiently} the new iteration.
Since node $i$ has knowledge only of $f_i$, writing  $F(\bx_i)=f_i(\bx_i)+\sum_{j\neq i}f_j(\bx_i)$, leads naturally to a convexification of  ${F}$ having the following form:  i) at every iteration $n$, the (possibly) nonconvex  $f_i(\bx_i)$ is replaced by a strongly convex surrogate, say  $\widetilde{f}_i(\bullet;\bx_i[n]):\mathcal K \rightarrow \mathbb{R}$, which may depend on the current iterate $\mathbf{x}_i[n]$; and  ii)  $\sum_{j\neq i}f_j(\bx_i)$  is  linearized around $\bx_i[n]$ (because it  is not  available at node $i$). More formally, the proposed  updating scheme reads: at every iteration $n$,  given the local  estimate $\bx_i[n]$, each agent $i$ solves the following \emph{strongly convex} optimization problem:\vspace{-0.1cm}
\begin{eqnarray}\label{best_resp_x_hat_2}
&&\hspace{-1.3cm}\widehat{\bx}_i(\bx_i[n])\nonumber\\
&&\hspace{-1.3cm}\triangleq{\rm arg\!}\min_{\!\!\!\!\!\!\!\!\!\!\!\!\!\boldsymbol{\bx_i}\in\mathcal{K}} \vspace{-0.2cm} {{\widetilde{f}_{i}(\bx_{i};\bx_{i}[n])+\boldsymbol{\pi}_{i}(\bx_{i}[n])^{T}(\bx_{i}-\bx_{i}[n])}}+G(\bx_i),
\end{eqnarray}
where  $\boldsymbol{\pi}_i(\bx_i[n])$ is the gradient of  $\sum_{j\neq i}f_j(\bx_i)$ at $\mathbf{x}_i[n]$, i.e.
\begin{equation}\label{pi}
\boldsymbol{\pi}_i(\bx_i[n])\triangleq\sum_{j\neq i}\nabla_{\bx}f_j(\bx_i[n]). \vspace{-0.1cm}
\end{equation}
Note that $\widehat{\bx}_i(\bx_i[n])$ is well-defined, because (\ref{best_resp_x_hat_2}) has a unique solution.   The idea behind the iterate (\ref{best_resp_x_hat_2}) is to compute stationary solutions of Problem (\ref{Problem}) as fixed-points of  the mappings $\widehat{\bx}_i(\bullet)$.  Postponing the convergence analysis  to Sec. \ref{sec:NEXT_algo}, a first natural question is about the choice  of the surrogate function $\widetilde{f}_i(\bullet;\bx_i[n])$.   The  next proposition addresses this issue and establishes the connection between the  fixed-points of  $\widehat{\bx}_i(\bullet)$  and the stationary solutions of Problem (\ref{Problem}); the proof follows the same steps as \cite[Prop. 8(b)]{Scutari-Facchinei-Sagratella} and thus is omitted. \vspace{-0.2cm}
\begin{proposition}\label{Prop:fixed-point-stationary} Given Problem (\ref{Problem}) under A1-A6, suppose that  $\widetilde{f}_i$ satisfies the following conditions:
 \begin{description}
\item[ (F1)]
 $\widetilde{f}_{i} (\mathbf{\bullet}; \mathbf{x})$ is uniformly strongly  convex with constant $\tau_i>0$ on $\mathcal K$;\smallskip
\item[  (F2)]  $\nabla \widetilde{f}_{i} (\mathbf{x};\mathbf{x}) = \nabla f_i(\mathbf{x})$ for all $\mathbf{x} \in \mathcal K$;\smallskip
\item[  (F3)]  $\nabla \widetilde{f}_{i} (\mathbf{x};\mathbf{\bullet})$ is uniformly Lipschitz continuous
on $ \mathcal K$.
\end{description}
Then, the set of fixed-point of $\widehat{\bx}_i(\bullet)$ coincides with that of the stationary solutions of  (\ref{Problem}). Therefore, $\widehat{\bx}_i(\bullet)$ has a fixed-point.
\end{proposition}\vspace{-0.2cm}
Conditions F1-F3 are quite natural: $\widetilde{f}_{i}$  should be regarded as a (simple) convex, local,
approximation of $f_i$ at the point $\bx$   that preserves the first order properties of $f_i$. Several feasible choices are possible for a given $f_i$; the  appropriate one depends on
the problem at hand and computational requirements; we
discuss alternative  options for $\widetilde{f}_i$ in Sec. \ref{sec:NEXT_discussion}. Here, we only remark   that no extra conditions  on $\widetilde{f}_i$ are required  to guarantee convergence of the proposed algorithms.}

\noindent \textbf{Step 2 (consensus update):} To force the asymptotic agreement among the $\bx_i$'s,  a consensus-based step is employed on $\widehat{\bx}_i(\bx_i[n])$'s. Each agent $i$ updates its  $\bx_i$ as:\vspace{-0.1cm}{
 \begin{equation}\label{consensus_update}
\bx_i[n+1]= \sum_{j\in \mathcal{N}_i^\text{in}[n]} w_{ij}[n]\, \widehat{\bx}_j(\bx_j[n]),\vspace{-0.1cm}
\end{equation}
where  $(w_{ij}[n])_{ij}$ is any set of (possibly time-varying) weights     satisfying Assumption B2; several choices are possible, see Sec. \ref{sec:NEXT_discussion} for details.   Since the weights are constrained by the network topology, \eqref{consensus_update} can be implemented via local message exchanges: agent $i$ updates its estimate $\bx_i$ by averaging over the  solutions $\widehat{\bx}_j(\bx_j[n])$ received   from its neighbors.  }
\\{ The rationale behind the proposed iterates (\ref{best_resp_x_hat_2})-(\ref{consensus_update}) is to compute fixed-points\,$\bx_i^{\infty}\,$of the mappings $\widehat{\bx}_i(\bullet)$ [i.e., $\widehat{\bx}_i(\bx_i^\infty)=\bx_i^\infty$ for all $i$] that are also stationary solutions of Problem (\ref{Problem}), while reaching asymptotic  consensus on  $\bx_i$, i.e., $\bx_i^{\infty}=\bx_j^{\infty}$, for all $i,j$, with $i\neq j$; this fact will be proved in Sec. \ref{sec:NEXT_algo}.}

\noindent {\textbf{Toward a fully distributed implementation:} The computation of $\widehat{\bx}_i(\bx_i[n])$ in (\ref{best_resp_x_hat_2}) is not fully distributed yet, because the evaluation of $\boldsymbol{\pi}_{i}(\bx_i[n])$ in (\ref{pi}) would require the knowledge of all $\nabla f_j(\bx_i[n])$, which is not available locally at node $i$.  To cope with this issue, the proposed approach consists in replacing $\boldsymbol{\pi}_i(\bx_i[n])$  in (\ref{best_resp_x_hat_2}) with a \emph{local} estimate, say  $\widetilde{\boldsymbol{\pi}}_{i}[n]$,  asymptotically converging to $\boldsymbol{\pi}_{i}(\bx_i[n])$, and solve instead}
\begin{align}\label{best_resp_x_hat_3}
&\widetilde{\bx}_i(\bx_i[n],\widetilde{\boldsymbol{\pi}}_{i}[n]) \nonumber\\
&\triangleq{\rm arg\!}\min_{\!\!\!\!\!\!\!\!\!\!\!\!\!\boldsymbol{\bx_i}\in\mathcal{K}} \underset{\triangleq\,\widetilde{U}_{i}(\bx_{i};\bx_{i}[n],\widetilde{\boldsymbol{\pi}}_i[n])}{\underbrace{\widetilde{f}_{i}(\bx_{i};\bx_{i}[n])+\widetilde{\boldsymbol{\pi}}_i[n]^T(\bx_{i}-\bx_{i}[n])+G(\bx_i)}}. \vspace{-0.2cm}
\end{align}
Rewriting $\boldsymbol{\pi}_i(\bx_i[n])$   in (\ref{pi}) as  \vspace{-.2cm}
\begin{equation}\label{pi2}
\boldsymbol{\pi}_i(\bx_i[n]) = \underset{\triangleq\overline{\nabla f}(\bx_{i}[n])}{I\cdot\underbrace{\left(\frac{1}{I}\sum_{j=1}^{I}\nabla f_{j}(\bx_{i}[n])\right)}}-\nabla f_{i}(\bx_{i}[n]),\vspace{-.1cm}
 \end{equation}
we propose to update $\widetilde{\boldsymbol{\pi}}_{i}[n]$ mimicking \eqref{pi2}:
\begin{equation}\label{pi3}
\widetilde{\boldsymbol{\pi}}_i[n]\triangleq I\cdot \by_i[n]-\nabla f_i(\bx_i[n]),
\end{equation}
where    $\by_i[n]$  is a local auxiliary variable (controlled by user $i$) that aims to asymptotically track  $\overline{\nabla f}(\bx_i[n])$. Leveraging \emph{dynamic}  average consensus methods \cite{Zhu-Martinez}, this can be done updating $\by_i[n]$ according to the following recursion:\vspace{-0.2cm}
\begin{equation}\label{y2}
\hspace{-0.04cm}\by_i[n+1]\triangleq\sum_{j=1}^I w_{ij}[n]\by_j[n] \hspace{-0.02cm} + \hspace{-0.02cm}\left(\nabla f_i(\x_i[n+1])\hspace{-0.02cm}-\hspace{-0.02cm}\nabla f_i(\bx_i[n])\right) \hspace{-0.2cm}\vspace{-0.05cm}
\end{equation}
with $\by_i[0]\triangleq\nabla f_i(\bx_i[0])$. In fact, if the sequences $\{\bx_i[n]\}_n$ are convergent and consensual, it holds $\left\Vert \by_{i}[n]-\overline{\nabla f}(\bx_{i}[n])\right\Vert \underset{n\rightarrow\infty}{\longrightarrow}0$ {(a fact that will be proved in the next section, see Theorem \ref{simplified_convergence_th})}, and thus $\left\Vert \widetilde{\boldsymbol{\pi}}_{i}[n]-\boldsymbol{\pi}_i(\bx_i[n])\right\Vert \underset{n\rightarrow\infty}{\longrightarrow}0$. Note that the update of $\by_i[n]$, and thus $\widetilde{\boldsymbol{\pi}}_i[n]$, can be now performed locally with message exchanges with the agents in the neighborhood $\mathcal N_i$.

\vspace{-0.3cm}
\subsection{The NEXT algorithm}\label{sec:NEXT_algo}
We are now in the position to formally introduce the NEXT algorithm, as given in Algorithm 1.  NEXT  builds on the iterates \eqref{best_resp_x_hat_3},  \eqref{consensus_update} (wherein each $\widehat{\bx}_j$ is replaced by $\widetilde{\bx}_j$) and \eqref{pi3}-\eqref{y2} introduced in the previous section. Note that in S.2, in addition to solving the strongly convex optimization problem (\ref{best_resp_x_hat_3}), we also introduced a step-size in the iterate: the new point $\bz_i[n]$ is a convex combination of the current estimate ${\bx}_i[n]$ and the solutions of (\ref{best_resp_x_hat_3}). Note that  we used the following simplified notation: $\widetilde{\bx}_i(\bx_i[n],\widetilde{\boldsymbol{\pi}}_{i}[n])$ in (\ref{best_resp_x_hat_3}) and $\nabla f_i(\bx_i[n])$ are denoted in Algorithm 1 as $\widetilde{\bx}_i[n]$ and $\nabla f_i[n]$, respectively. The convergence properties of NEXT are given in Theorem \ref{simplified_convergence_th}.

\begin{algorithm}[t]

$\textbf{Data}:$ $\bx_{i}[0]\in \mathcal{K}$, $\by_i[0]= \nabla f_i[0]$,    $\widetilde{\boldsymbol{\pi}}_{i}[0]=I\by_i[0]-\nabla f_i[0]$, $\forall i=1,\ldots ,I$, and $\{\bW[n]\}_n$. Set $n=0$.\smallskip

\texttt{$\mbox{(S.1)}$}$\,\,$If $\mathbf{x}[{n}]$ satisfies a termination
criterion: STOP;\smallskip

\texttt{$\mbox{(S.2)}$} \texttt{Local SCA optimization}:  Each agent $i$

\hspace{1.1cm} \vspace{-0.05cm} (a) computes locally $\widetilde{\bx}_{i}[n]$:
\begin{equation}\label{opt_prob_alg}
\widetilde{\bx}_{i}[n]\triangleq \underset{\bx_i\in\mathcal{K}}{\text{argmin}} \,\,\widetilde{U}_i\left(\bx_i;\bx_i[n],\widetilde{\boldsymbol{\pi}}_{i}[n]\right) \nonumber
\end{equation}

\hspace{1.1cm} \vspace{-0.05cm} (b) updates its local variable $\bz_i[n]$:
\begin{equation}
\bz_i[n]=\bx_i[n]+\alpha[n]\left(\widetilde{\bx}_{i}[n]-\bx_i[n]\right) \nonumber
\end{equation}

\vspace{-0.05cm}
\texttt{$\mbox{(S.3)}$} \texttt{Consensus update}:   Each agent $i$ collects data from its current neighbors  and updates $\bx_i[n]$, $\by_i[n]$, and $\widetilde{\boldsymbol{\pi}}_{i}[n]$:\vspace{-.2cm}{
\begin{align}
  &\hbox{(a)$\;\;  \displaystyle \bx_i[n+1]= \sum_{j=1}^I w_{ij}[n]\, \bz_j[n]$}\nonumber\\
  &\hbox{(b)$\;\;  \displaystyle \by_i[n+1]=\sum_{j=1}^I w_{ij}[n]\,\by_j[n]+\left(\nabla f_i[n+1]-\nabla f_i[n]\right)$}\nonumber\\
  &\hbox{(c)}\;\;  \widetilde{\boldsymbol{\pi}}_{i}[n+1]=I\cdot \by_i[n+1]-\nabla f_i[n+1]\nonumber
\end{align}}

\vspace{-0.4cm}
\texttt{$\mbox{(S.4)}$} $n\leftarrow n+1$, and go to \texttt{$\mbox{(S.1)}.$}

\protect\caption{\hspace{-2.5pt}\textbf{:} \label{alg:general}\textbf{in-Network succEssive conveX approximaTion (NEXT)}}
\end{algorithm}

\vspace{-0.2cm}
\begin{theorem}\label{simplified_convergence_th}{
Let $\{\mathbf{x}[n]\}_n\triangleq \{(\mathbf{x}_i[n])_{i=1}^I\}_n$ be the sequence generated by Algorithm 1, and let $\{\overline{\mathbf{x}}[n]\}_n\triangleq \{(1/I)\,\sum_{i=1}^I\mathbf{x}_i[n]\}_n$ be its average. Suppose that i) Assumptions A and B hold; and ii) the step-size sequence $\{\alpha[n]\}_n$ is chosen so that  $\alpha[n]\in (0,1]$, for all $n$,
\begin{equation}\label{step-size}
   \hbox{$\sum_{n=0}^{\infty}\alpha[n]=\infty$ \hspace{.2cm} and \hspace{.2cm} $\sum_{n=0}^{\infty}\alpha[n]^2<\infty.$}\\
\end{equation}
Then,\,(a)\,\emph{\texttt{[convergence]}:\,}the sequence\,$\{\overline{\mathbf{x}}[n]\}_n$\,is bounded and all its limit points are stationary solutions of Problem  (\ref{Problem}); (b) \emph{\texttt{[consensus]}}:   all the sequences  $\{\bx_i[n]\}_n$  asymptotically agree, i.e., $\|\mathbf{x}_{i}[n]-\overline{{\mathbf{x}}}[n]\|\underset{n\rightarrow\infty}{\longrightarrow}0
 $, for all $i=1,\ldots ,I$.}\vspace{-0.2cm}
\end{theorem}
\begin{proof}
See Appendix.
\end{proof}

\vspace{-0.2cm}
\indent Theorem \ref{simplified_convergence_th} states two results. First, the average estimate $\{\overline{\mathbf{x}}[n]\}_n$ converges to the set $\mathcal S$ of stationary solutions of (\ref{Problem}). Second, a consensus is asymptotically
achieved among the local estimates $\mathbf{x}_i[n]$. Therefore, the sequence $\{\mathbf{x}[n]\}_n$ converges to the set $\{\mathbf{1}\otimes{\overline{\mathbf{x}}}\,:\,{\overline{\mathbf{x}}}\in \mathcal S\}$.  In particular,  if $F$ is convex, Algorithm 1 converges (in the aforementioned sense) to the set of global optimal solutions of the resulting convex  problem   (\ref{Problem}). However, as already remarked, our result is more general and does not rely on the convexity of $F$. \vspace{-0.35cm}

\subsection{Design of the free parameters} \label{sec:NEXT_discussion}
NEXT represents the first family of  \emph{distributed  SCA} methods for Problem (\ref{Problem}). It is very general and  encompasses a
gamut of novel algorithms, each corresponding to various
forms of the approximant  $\widetilde{f}_i$, the weight matrices $\bW[n]$,
and the step-size sequence $\alpha[n]$, but \emph{all converging under the same conditions}. We outline next some effective choices for the aforementioned parameters.

\noindent{\textbf{On the choice of the surrogates $\widetilde{f}_i$:}} Adapting to our setting the approximation functions introduced in \cite{Scutari-Facchinei-Song-Palomar-Pang,Scutari-Facchinei-Sagratella}, the following examples are instances of $\widetilde{f}_i$ satisfying F1-F3.

\noindent $\bullet$ When $f_i$ has no special structure to exploit, the most obvious choice for $\widetilde{f}_i$  is the linearization of $f_i$ at $\bx_i[n]$:
\begin{align}\label{gradient_surrogate}
\widetilde{f}_i(\bx_i; \bx_i[n]) =\;\; &f_i(\bx_i[n]) + \nabla f_i(\bx_i[n])^T(\bx_i - \bx_i[n]) \nonumber\\
&+\dfrac{\tau_i}{2} \|\bx_i - \bx_i[n]\|^2,
\end{align}
where $\tau_i$ is any positive constant. The proximal regularization guarantees that $\widetilde{f}_i$  is strongly convex. The above surrogate  is essentially a reminiscence of the approximation of the objective function used in proximal-gradient algorithms. Note however that standard proximal-gradient algorithms are not directly applicable to Problem (\ref{Problem}), as they are not distributed.
\\\noindent $\bullet$ At another extreme, if $f_i$ is convex,  one could just take
\begin{equation}\label{cvx_preserved_surrogate}  \widetilde{f}_i(\bx_i; \bx_i[n]) = f_i(\bx_i) + \dfrac{\tau_i}{2} \|\bx_i - \bx_i[n]\|^2,
\end{equation}
with $\tau_i\geq 0$ ($\tau_i$ can be set to zero if $f_i$ is  strongly convex).
Differently from \eqref{gradient_surrogate}, this choice preserves the structure of $f_i$.
\\\noindent $\bullet$ Between the two ``extreme'' solutions proposed above, one can consider ``intermediate'' choices. For example,
 if   $f_i$ is convex, mimicking Newton schemes, one can take  $\widetilde{f}_i$ as a second order approximation of $f_i$, i.e.,
\begin{equation}\label{newton_surrogate}
\begin{array}{ll}
\widetilde{f}_i(\bx_i; \bx_i[n]) = f_i(\bx_i[n]) + \nabla f_i(\bx_i[n])^T(\bx_i - \bx_i[n]) \smallskip\\ \hspace{1cm}+
\displaystyle \frac{1}{2} (\bx_i - \bx_i[n])^T \nabla^2 f_i(\x_i[n]) (\bx_i - \bx_i[n]).
\end{array}
\end{equation}
\noindent $\bullet$ Another ``intermediate'' choice relying on a specific structure of each $f_i$ that has
important applications is the following. Suppose that $f_i$ is convex only in some components of $\bx_i$; let us split $\bx_i\triangleq (\bx_{i,1}, \bx_{i,2})$ so that $f_i(\bx_{i,1}, \bx_{i,2})$ is convex in $\bx_{i,1}$ for every $\bx_{i,2}$ such that $(\bx_{i,1}, \bx_{i,2})\in \mathcal K$, but not in $\bx_{i,2}$.  A natural choice for $\widetilde{f}_i$ is then: given $\bx[n]\triangleq (\bx_{i,1}[n], \bx_{i,2}[n])$,
\begin{equation}\label{partial_cvx_surrogate}
\begin{array}{ll}
 \hspace{-0.5cm}\widetilde{f}_i(\bx_i; \bx_i[n]) =  \widetilde{f}_i^{(1)}(\bx_{i,1}; \bx_{i,2}[n])+\dfrac{\tau_i}{2} \|\bx_{i,2} - \bx_{i,2}[n]\|^2\hspace{-0.4cm} \smallskip\\ \hspace{1cm}+\nabla_{x_{i,2}} f_i(\bx_i[n])^T (\bx_{i,2}-\bx_{i,2}[n])\hspace{-0.05cm}
 \end{array}
\end{equation}
where $\widetilde{f}_i^{(1)}(\bullet; \bx_{i,2}[n])$ is any  function still satisfying F1-F3 (written now in terms of  $\bx_{i,1}$ for given $\bx_{i,2}$). Any of the choices in (\ref{gradient_surrogate})-(\ref{newton_surrogate}) are valid for $\widetilde{f}_i^{(1)}(\bullet; \bx_{i,2}[n])$.
The rationale behind  (\ref{partial_cvx_surrogate}) is to preserve the favorable convex part of $f_i$ with respect to $\x_{i,1}$ while linearizing the nonconvex part.
\\\noindent $\bullet$ Consider the case where ${f}_i$ is block-wise convex but not convex on $\bx_i$. Let us assume that ${f}_i$ is convex in the two block-variables $\bx_{i,1}$ and  $\bx_{i,2}$, forming a partition $\bx_i=(\bx_{i,1}, \bx_{i,2})$, but not jointly (the case of more than two blocks can be similarly considered). Then, a natural choice for $\widetilde{f}_i$ is \begin{align}\label{block-wise_cvx_surrogate}
 \widetilde{f}_i(\bx_i; \bx_i[n]) =\;  &f_i(\bx_{i,1}, \bx_{i,2}[n])+ f_i(\bx_{i,1}[n], \bx_{i,2}) \nonumber\\
     &+\dfrac{\tau_i}{2} \|\bx_{i} - \bx_{i}[n]\|^2.
\end{align}
Note that, in the same spirit of the previous example,  instead of  $f_i(\bullet, \bx_{i,2}[n])$ and $f_i(\bx_{i,1}[n], \bullet)$ one can use any surrogate satisfying F1-F3 in the intended variables.
\\\noindent $\bullet$ As last example,  consider the case where $f_i$ is the composition of two functions, i.e., $f(\bx_i)=g(h(\bx_i))$, where   $g:\mathbb R \rightarrow \mathbb R$ is convex. Then,  a possible choice for $\tilde{f}_i$ is to preserve the convexity of $g$, while linearizing   $h$, resulting in the following surrogate \vspace{-0.1cm}
\begin{align}\label{composition_surrogate}
 &\widetilde{f}_i(\bx_i; \bx_i[n]) =\;  g\big(h(\bx_i[n])+\nabla h(\bx_i[n])^T(\bx_i-\bx_i[n])\big)\nonumber\\
 &\hspace{2cm}+\dfrac{\tau_i}{2} \|\bx_{i} - \bx_{i}[n]\|^2. \vspace{-0.1cm}
\end{align}
The above idea can be readily extended to the case where the inner function is a vector valued function.\\
\noindent \textbf{Distributed and parallel computing:} When each node is equipped with a multi-core  architecture or a cluster computer  (e.g., each node is a cloud), the proposed framework permits,  throughout a proper   choice of the surrogate functions, to distribute the computation   of the solution   of each subproblem (\ref{best_resp_x_hat_3}) across the cores.  To elaborate,   suppose that there are  $C$ cores available at each node $i$, and partition  $\bx_i=(\bx_{i,c})_{c=1}^C$ in $C$ (nonoverlapping) blocks, each of them subject to individual constraints only, i.e., $\bx_{i}\in \mathcal K\Leftrightarrow \bx_{i,c}\in \mathcal K_c$, for all $c=1,\ldots ,C$, with each $\mathcal K_c$ being closed and convex. Assume, also, that $G$ is block separable, i.e., $G(\bx)=\sum_{c=1}^C G_{i,c}(\bx_{i,c})$; an example of such a $G$ is the $\ell_1$-norm or the $\ell_2$-block norm. Then, choose    $\widetilde{f}_i$ as   additively separable     in the blocks $(\bx_{i,c})_{c=1}^C$, i.e.,   $\widetilde{f}_i(\bx_i; \bx_i[n])=\sum_{c=1}^C \widetilde{f}_{i,c}(\bx_{i,c};\bx_{i,-c}[n])$, where   each $\widetilde{f}_{i,c}(\bullet; \bx_{i,-c}[n])$ is any  surrogate  function satisfying F1-F3 in the variable $\bx_{i,c}$, and $\bx_{i,-c}[n]\triangleq (\bx_{i,p}[n])_{1=p\neq c}^C$ denotes the tuple of all blocks excepts the $c$-th one. For instance, if $f_i$ is strongly convex in each block $\bx_{i,c}$, one can choose $\widetilde{f}_{i,c}(\bx_{i,c}; \bx_{i,-c}[n])={f}_{i,c}(\bx_{i,c}; \bx_{i,-c}[n])$ [cf. (\ref{block-wise_cvx_surrogate})]. With the above choices,   the resulting problem (\ref{best_resp_x_hat_3}) becomes decomposable in  $C$ separate  strongly convex subproblems
\begin{equation}\label{parallel_computation}
 \min_{\bx_{i,c}\in\mathcal{K}_c}  \widetilde{f}_{i,c}(\bx_{i,c};\bx_{i,-c}[n])+\widetilde{\boldsymbol{\pi}}_{i,c}[n]^T(\bx_{i,c}-\bx_{i,c}[n])+G_{i,c}(\bx_{i,c}),
\end{equation}
for $c=1,\ldots, C$, where $\widetilde{\boldsymbol{\pi}}_{i,c}[n]$ denotes the $c$-th bock of $\widetilde{\boldsymbol{\pi}}_{i}[n]$. Each subproblem \eqref{parallel_computation} can be  now solved independently  by a  different core. It is interesting to observe that   the aforementioned instance of NEXT represents a  \emph{distributed} (across the nodes) \emph{and  parallel} (inside each node) solution method  for Problem (\ref{Problem}) (under the setting described above). To the best of our knowledge, this is the first nongradient-like scheme enjoying such a desirable feature.\\
\noindent{\textbf{On the choice of $\alpha[n]$ and $\bW[n]$:}} Conditions (\ref{step-size}) in Theorem \ref{simplified_convergence_th} on the step-size sequence $\{ \alpha[n]\}_n$ ensure that the step-size decays to zero, but not too fast. There are many diminishing step-size rules in the literature satisfying  (\ref{step-size}); see, e.g., \cite{Bertsekas2}. For instance, we found the following two rules very effective in our experiments \cite{Scutari-Facchinei-Song-Palomar-Pang}:
\begin{align}
&\texttt{Rule 1:} \;\;\; \alpha[n]=\frac{\alpha_0}{(n+1)^{\beta}}, \;\; \alpha_0>0, \;\; 0.5<\beta\leq1, \label{step1}\\
&\texttt{Rule 2:} \;\;\; \alpha[n]=\alpha[n-1](1-\mu\alpha[n-1]), \quad n\geq1, \label{step2}
\end{align}
with $\alpha[0]\in(0,1]$ and $\mu\in(0,1)$.
Notice that while these rules
may still require some tuning for optimal behavior, they are quite reliable, since in general we are not
using a (sub)gradient direction, so that many of the well-known practical drawbacks associated
with a (sub)gradient method with diminishing step-size are mitigated in our setting.  Furthermore, this choice of step-size does not require any form of centralized coordination, which is  a key feature in our distributed  environment.\\
\indent The weight matrices $\bW[n]$ need to satisfy the \textit{doubly stochasticity} condition (\ref{Doubly_Stochastic}). Several choices have been proposed in the literature, such as the uniform weights \cite{Blondel}; the Laplacian weights \cite{Scherber}; the maximum degree weight, the Metropolis-Hastings, and the least-mean square consensus weight rules \cite{Xiao}; and the relative degree(-variance) rule \cite{Cattivelli-Sayed}, to name a few. On a practical side, the above rules call for specific protocols and signaling among  nodes to be implemented.  In fact, while right-stochasticity (i.e., $\bW[n]\mathbf{1}=\mathbf{1}$)  can be easily  enforced even in the case of time-varying topologies [at every iteration, each agent can discriminate the weights $(w_{ij}[n])_{j\in \mathcal N_i^{\text{in}}\setminus {i}}$ based on the  packets sent by its  neighbors and successfully received],   left-stochasticity (i.e., $\mathbf{1}^T\bW[n]=\mathbf{1}^T$) is more difficult to enforce and requires some coordination among neighboring agents in the choice of the weights  forming the columns of  $\bW[n]$. The design and analysis of such broadcast communication protocols go beyond the scope of this paper; we refer to  \cite{Benezit} for a practical implementation of  broadcast/gossip strategies and  consensus protocols.

\noindent {\textbf{NEXT vs. gradient-consensus algorithms.} The following question arises naturally from the above discussion: How does NEXT compare with classical gradient-consensus algorithms (e.g., \cite{Nedic-Ozdaglar-Parillo,Bianchi-Jakubowicz}), when  each $\tilde{f}_i$ is chosen as in \eqref{gradient_surrogate} (i.e., a full linearization of the agents' objectives is used as surrogate function)? We term such an instance of NEXT, \emph{NEXT linearization} (NEXT-L). We address the question considering, for simplicity,  the instance of Problem (\ref{Problem}) wherein $G=0$ and under time-invariant topology. The main iterate of classical consensus scheme reads \cite{Bianchi-Jakubowicz}: given $(\bx_i[n])_{i=1}^I$,  \vspace{-0.1cm}
\begin{align}
& \mathbf{z}_i[n]= \boldsymbol{\Pi}_{\mathcal{K}} \left(\bx_i[n]-\alpha[n] \nabla f_i(\bx_i[n]) \right),\label{grad-update}\medskip\\
& \bx_i[n+1]= \sum_{j=1}^I w_{ij}\, \bz_j[n],\label{grad-update2}\
\end{align}
for all $i=1,\ldots, I$, where $\boldsymbol{\Pi}_{\mathcal{K}}(\cdot)$ denotes the projection onto  the (convex and closed) set $\mathcal{K}$. Using \eqref{gradient_surrogate}, it is not difficult to see that Step 2 of NEXT can be equivalently rewritten as
\begin{align}
&\widetilde{\mathbf{x}}_i[n]= \boldsymbol{\Pi}_{\mathcal{K}} \left(\bx_i[n]-\dfrac{1}{\tau_i}\, \left(\nabla f_i(\bx_i[n]) + \widetilde{\boldsymbol{\pi}}_i[n]\right) \right),\label{NEXT-update_comparison1}\medskip\\
& \bz_i[n]= \bx_i[n]+\alpha[n](\widetilde{\bx}_i[n]-\bx_i[n]), \label{NEXT-update_comparison2}
\end{align}
with  $\Vert \widetilde{\boldsymbol{\pi}}_{i}[n]-\sum_{j\neq i} \nabla f_j(\bx_i[n])\Vert \underset{n\rightarrow\infty}{\longrightarrow}0$ (as a consequence of the proof of Theorem \ref{simplified_convergence_th}). Comparing (\ref{grad-update}) with (\ref{NEXT-update_comparison1})-(\ref{NEXT-update_comparison2}), one can infer that, besides minor differences (e.g., the step-size rules),  the gradient-consensus scheme and NEXT-L  update the local variables  $\bz_i$ using different directions, $\mathbf{z}_i[n]-\bx_i[n]$ and  $\widetilde \bx_i[n]-\bx_i[n]$, respectively.
The former is based only on the gradient of the local function $f_i$, whereas the latter retains some information, albeit inexact,  on the gradient of the whole sum-utility $F$ (through $\widetilde{\boldsymbol{\pi}}_{i}$); this information becomes more and more accurate as the iterations go. This better exploitation of the sum-utility comes however at the cost of an extra consensus step: at each iteration, NEXT-L requires  twice the communication of the gradient-consensus scheme [compare Step 3 of Algorithm 1 with  (\ref{grad-update2})]. Our experiments show that overall the extra information on the gradient of the sum-utility used in NEXT-L can significantly enhance the practical convergence of the algorithm (see, e.g.,  Fig. 1, Sec. IV.A):  NEXT-L requires significantly less signaling than gradient schemes to reach the same   solution accuracy.

\vspace{-0.2cm}
\section{Inexact NEXT}\label{sec:extensions}

In many situations (e.g., in the case
of large-scale problems), it can be useful to further reduce
the computational effort  to solve the subproblems in (\ref{best_resp_x_hat_3}) by allowing inexact computations $\bx_i^{\texttt{inx}}[n]$ of  $\widetilde{\bx}_i[n]$ in Step 2(a) of Algorithm 1, i.e.,
\begin{equation}\label{precision}
\|\bx_i^{\texttt{inx}}[n]- \widetilde{\bx}_i[n]\|\leq \varepsilon_i[n],
\end{equation}
where $\varepsilon_i[n]$   measures the accuracy in computing the solution. This is a noticeable feature of the proposed algorithm that allows to control the cost per iteration without affecting too much, experience shows, the empirical convergence speed.\\ \indent
The generalization of Algorithm 1 including inexact updates  is described in Algorithm  \ref{alg:general}, and is termed  Inexact NEXT; its convergence is stated  in Theorem \ref{main_th}.

\vspace{-.2cm}
\begin{theorem}\label{main_th}
Let $\{\mathbf{x}[n]\}_n\triangleq \{(\mathbf{x}_i[n])_{i=1}^I\}_n$ be the sequence generated by Algorithm 2, and let $\{\overline{\mathbf{x}}[n]\}_n\triangleq \{(1/I)\,\sum_{i=1}^I\mathbf{x}_i[n]\}_n$ be its average, under the setting of Theorem \ref{simplified_convergence_th}. Choose   the step-size sequence $\{\alpha[n]\}_n$ so that, in addition to conditions in Theorem \ref{simplified_convergence_th}, the following holds
\begin{equation}\label{errors}
\hbox{$\sum_{n=0}^{\infty}\alpha[n]\,\varepsilon_i[n]<\infty,$\hspace{.4cm} $\forall i$.}
\end{equation}
Then, statements (a) and (b) of  Theorem  \ref{simplified_convergence_th}  hold.\vspace{-.2cm}
\end{theorem}
\begin{proof}
See Appendix.
\end{proof}

\vspace{-.1cm}

\begin{algorithm}[t]
$\textbf{Data}:$ Same as in Algorithm 1, and $\{\varepsilon_i[n]\}_n$.\smallskip

Same steps as in Algorithm 1, with Step 2 replaced by

\texttt{$\mbox{(S.2)}$}\texttt{Local inexact SCA}:  Each agent $i$

\hspace{0.8cm} (a) solves (\ref{best_resp_x_hat_3}) with accuracy $\varepsilon_i[n]$: Find a $\bx_i^{\texttt{inx}}[n]\in \mathcal{K}$

\hspace{1.4cm}  s.t. \vspace{-0.4cm}
 $$\|\widetilde{\bx}_{i}[n]-\bx_i^{\texttt{inx}}[n]\|\leq \varepsilon_i[n];$$

\hspace{0.8cm} (b) updates its local variable $\bz_i[n]$:
\begin{equation}\label{z_update}
\bz_i[n]=\bx_i[n]+\alpha[n]\left(\bx_i^{\texttt{inx}}[n]-\bx_i[n]\right) \nonumber
\end{equation}

\protect\caption{\hspace{-2.5pt}\textbf{:} \label{alg:general}\textbf{Inexact NEXT}}
\end{algorithm}

\vspace{-0.1cm}
As expected, in the presence of errors, convergence of Algorithm \ref{alg:general} is guaranteed if the sequence of approximated
problems in S.2(a) is solved with increasing accuracy. Note that, in addition to require $\varepsilon_i[n]\rightarrow 0$, condition $\sum_{n=0}^{\infty}\alpha[n]\,\varepsilon_i[n]<\infty$ of Theorem \ref{main_th} imposes also a constraint on the rate by which $\varepsilon_i[n]$ goes to zero, which depends on the rate of decrease of  $\alpha[n]$. An example of error sequence satisfying the above condition is $\varepsilon_i[n]\leq c_i \alpha[n]$,  where $c_i$ is any finite positive constant \cite{Scutari-Facchinei-Sagratella}. Interesting,
such a condition can be forced in Algorithm \ref{alg:general} in a distributed way, using classical error bound results in convex analysis; see, e.g., \cite[Ch. 6, Prop. 6.3.7]{Facchinei}.

\vspace{-0.3cm}

\section{Applications and Numerical Results}

In this section, we customize the proposed algorithmic framework to specific applications in several areas, namely: Signal Processing, Communications, and Networking. Applications include i) cooperative target localization; ii) distributed spectrum cartography in cognitive radio (CR) networks; iii) flow control in communication networks; and iv) sparse distributed estimation in wireless sensor networks.    Numerical results show that NEXT compares favorably with respect to current ad-hoc schemes.  %while reaching the same (stationary) solutions.

\vspace{-0.45cm}
\subsection{Distributed target localization}

Consider a  multi-target localization problem: a sensor network  aims to locate $N_T$ common targets, based on some distance measurements.  Let $\bx_t\in \mathbb{R}^p$,  with  $p=2,3$,  denote the  vector (to be determined) of the (2D or 3D) coordinates of each target $t=1,\ldots, N_T$ in the network (with respect to  a global reference system).  Each node knows its position $\boldsymbol{\omega}_i$ and has access to noisy measurements $\varphi_{it}$ of the squared distance to each target. Thus, the least squares estimator for the target positions $\bx\triangleq (\bx_t)_{t=1}^{N_T}$ is the solution of the following \emph{nonconvex} optimization problem \cite{Chen-Sayed}: \vspace{-0.2cm}
\begin{equation}\label{Localization_Problem}
\min_{\bx\in \mathcal{K}}\, F(\bx)\triangleq \sum_{i=1}^{I}\sum_{t=1}^{N_T}\big(\varphi_{it}-\|\bx_t-\boldsymbol{\omega}_i\|^2\big)^2,\vspace{-0.1cm}
\end{equation}
where $\mathcal{K}$ is a compact and convex set modeling geographical bounds on the position of the targets. Problem (\ref{Localization_Problem}) is clearly an instance of Problem  (\ref{Problem}), with $G(\bx)=0$ and\vspace{-0.2cm}
\begin{equation}\label{Loc_fi}
f_i(\bx)=\sum_{t=1}^{N_T}(\varphi_{it}-\|\bx_t-\boldsymbol{\omega}_i\|^2)^2.\vspace{-0.1cm}
\end{equation}
{Several choices for the surrogate function $\widetilde{f}_i(\bx;\bx[n])$ are possible (cf. Sec. \ref{sec:NEXT_discussion}); two instances are given next. Since (\ref{Loc_fi}) is a fourth-order polynomial in each $\bx_t$, %a first choice  for $\widetilde{f}_i(\bx;\bx[n])$ is to
we might preserve the ``partial'' convexity in $f_i$  by  keeping the first and second order (convex) terms in each summand of (\ref{Loc_fi}) unaltered and linearizing the higher order terms. This leads to
\begin{equation}\label{Loc_fi_tilde}\vspace{-0.1cm}
\widetilde{f}_i(\bx;\bx[n])=\sum_{t=1}^{N_T}\widetilde{f}_{it}(\bx;\bx[n])+\frac{\tau}{2}\|\bx_t-\bx_t[n]\|^2\vspace{-0.1cm}
\end{equation}
where $\widetilde{f}_{it}(\bx;\bx[n])\triangleq \bx_t^T \bA_i\bx_t - \bb_{it}[n]^T(\bx_t-\bx_t[n])$, with
$\bA_i \triangleq\;  4\boldsymbol{\omega}_i\boldsymbol{\omega}_i^T+2\|\boldsymbol{\omega}_i\|^2\bI_p$, and
 $
\bb_{it}[n]  \triangleq\;4\|\boldsymbol{\omega}_i\|^2\boldsymbol{\omega}_i-4(\|\bx_t[n]\|^2-\varphi_{it})(\bx_t[n]-\boldsymbol{\omega}_i) +8(\boldsymbol{\omega}_i^T\bx_t[n])\bx_t[n].$\\\indent
A second option is of course to linearize the whole $f_i(\bx)$, which leads to the following surrogate function:
\begin{equation}\label{Loc_fi_tilde2}
\widetilde{f}_i(\bx;\bx[n])=\nabla_{\bx}f_i(\bx[n])^T(\bx-\bx[n])+\frac{\tau}{2}\|\bx-\bx[n]\|^2,
\end{equation}
with $\nabla_{\bx}f_i(\bx[n])=(\nabla_{\bx_t}f_i(\bx_t[n]))_{t=1}^{N_T}$, and $\nabla_{\bx_t}f_i(\bx_t[n])=-4(\varphi_{it}-\|\bx_t[n]-\boldsymbol{\omega}_i\|^2)(\bx_t[n]-\boldsymbol{\omega}_i)$.}\smallskip

\begin{figure}[t]
\centering
\includegraphics[width=6.4cm]{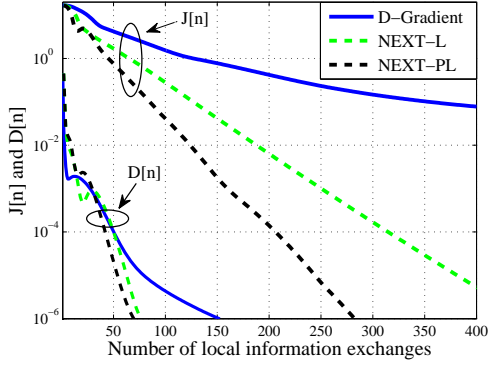}
\caption{ Distributed target localization: Distance from stationarity ($J[n]$) and consensus disagreement ($D[n]$)  versus number of local information exchanges. Both versions of NEXT are significantly faster than D-gradient.}\label{Loc_comparison_a}\vspace{-0.5cm}
\end{figure}
\noindent{\textbf{Numerical Example:}} {We simulate a time-invariant connected network composed of  $I=30$ nodes, randomly deployed over the unitary square $[0,1]\times[0,1]$}. We consider $N_T=3$ targets with positions $(0.03, 0.85), (0.86, 0.5), (0.6, 0.01)$.
Each measurement $\varphi_{it}$ [cf.  (\ref{Localization_Problem})] is corrupted by additive white Gaussian noise with zero mean, and variance chosen such that the minimum signal to noise ratio (SNR) on the measurements taken over the network is equal to -20 dB. The set $\mathcal{K}$ is chosen as $\mathcal{K}=[0,1]^2$ (the  solution is in the unitary square).
{We tested two instances of NEXT, namely: the one  based on the surrogate functions in \eqref{Loc_fi_tilde}, termed \emph{NEXT partial linearization} (NEXT-PL); and ii) the one based on \eqref{Loc_fi_tilde2}, termed \emph{NEXT  linearization} (NEXT-L). Note that, in the above setting, one can compute the  best-response $\widetilde{\bx}_{i}[n]=(\widetilde{\bx}_{it}[n])_{t=1}^{N_t}$  in S.2(a) [cf. (\ref{best_resp_x_hat_3})] in closed form, as outlined next. Considering the surrogate \eqref{Loc_fi_tilde} (similar argument applies  to $\widetilde{\bx}_{i}[n]$ based on the surrogate function \eqref{Loc_fi_tilde2}, we omit the details because of space limitations),} and defining
\begin{equation}\label{loc_unc_solution}
\widehat{\bx}_{it}[n]\triangleq(\bA_i+\tau\bI_p)^{-1}(\bb_{it}[n]-\widetilde{\boldsymbol{\pi}}_i[n]+\tau\bx_{it}[n]),
\end{equation}
for all $t=1,\ldots,N_T$, and partitioning $\widehat{\bx}_{it}[n]=\big[\widehat{x}_{it}^{(1)}[n],\widehat{x}_{it}^{(2)}[n]\big]^T$, we have
\vspace{-0.1cm}
\begin{equation}
\widetilde{\bx}_{it}[n]=\left\{
                 \begin{array}{lllll}
                   (\widehat{x}_{it}^{(1,0)}[n],0)^T, & \hspace{-2cm}\hbox{if $\widehat{x}_{it}^{(1)}[n]\in [0,1],\widehat{x}_{it}^{(2)}[n]< 0$;} \\
                   (\widehat{x}_{it}^{(1,1)}[n],1)^T, & \hspace{-2cm}\hbox{if $\widehat{x}_{it}^{(1)}[n]\in [0,1],\widehat{x}_{it}^{(2)}[n]> 1$;} \\
                   (0,\widehat{x}_{it}^{(2,0)}[n])^T, &\hspace{-2cm} \hbox{if $\widehat{x}_{it}^{(1)}[n]<0,\widehat{x}_{it}^{(2)}[n]\in [0,1]$;} \\
                   (1,\widehat{x}_{it}^{(2,1)}[n])^T, & \hspace{-2cm}\hbox{if $\widehat{x}_{it}^{(1)}[n]>1,\widehat{x}_{it}^{(2)}[n]\in [0,1]$;} \\
                \left(\left[\widehat{x}_{it}^{(1)}[n]\right]_{0}^{1},\,\left[\widehat{x}_{it}^{(2)}[n]\right]_{0}^{1}\right)^{T}, & \hbox{otherwise;} \\
                \end{array}
                        \right. \nonumber
\end{equation}
$[x]_0^1\triangleq \min(\max(x,0),1)$ and
\begin{align}
\widehat{x}_{it}^{(1,j)}[n]&=\underset{x}{\text{argmin}}\; \widetilde{f}_{it}\big((x,j);\bx_{it}[n]\big), \nonumber\\
\widehat{x}_{it}^{(2,j)}[n]&=\underset{x}{\text{argmin}}\; \widetilde{f}_{it}\big((j,x);\bx_{it}[n]\big), \nonumber
\end{align}
with $j=0,1$, are the minima of one-dimensional quadratic functions, which can  be easily computed in closed form. Note that the quantity (\ref{loc_unc_solution}) involves the inversion of matrix $\bA_i+\tau\bI_p$, which is   independent of the iteration index $n$, and thus needs to be  computed only once.
\begin{figure}[t]
\centering
\includegraphics[width=6.5cm]{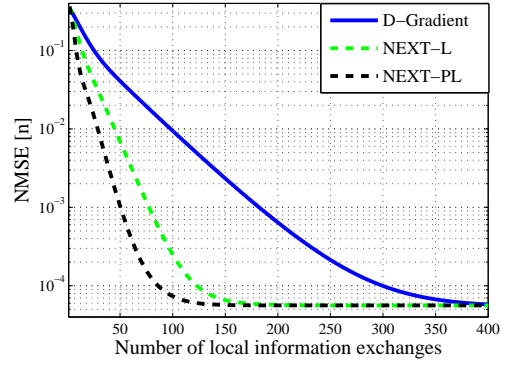}
\caption{Distributed target localization: Normalized MSE (NMSE$[n]$)  versus number of local information exchanges.}\label{Loc_comparison_b}\vspace{-0.3cm}
\end{figure}
We  compared  NEXT, in the two versions, NEXT-PL and NEXT-L, with the only algorithm available in the literature with provable convergence for Problem (\ref{Localization_Problem}), i.e., the distributed gradient-based method  in \cite{Bianchi-Jakubowicz}, which we term  \emph{Distributed-gradient} (D-Gradient). The tuning of the algorithms is the following. They are all initialized uniformly at random  within the unitary square, and use the step-size rule (\ref{step2}) with $\alpha[0]=0.1$, $\mu=0.01$, $\tau=10$, for the two versions of NEXT, and $\alpha[0]=0.05$ and $\mu=0.05$ for the D-Gradient. According to our experiments, the above choice empirically results in the best practical convergence of all methods. The weights in (\ref{weights}) are chosen time-invariant and satisfying the Metropolis rule for all algorithms \cite{Xiao}. We measured the progress of the algorithms  using  the following two merit functions:\vspace{-0.1cm}
\begin{equation}\label{Metric_J}
\begin{array}{ll}
J[n]\triangleq  \Big\|\bar{\bx}[n]-\boldsymbol{\Pi}_{\mathcal{K}}\big(\bar{\bx}[n]-\nabla_{\bx} F(\bar{\bx}[n])\big)\Big\|_{\infty}\smallskip\\
D[n]\triangleq \dfrac{1}{I}\sum_i\|\bx_i[n]-\bar{\bx}[n]\|^2,\vspace{-0.1cm}
\end{array}
\end{equation}
where $\boldsymbol{\Pi}_{\mathcal{K}}(\cdot)$ denotes the projection onto $\mathcal{K}$,  $\bar{\bx}[n]\triangleq (1/I)\sum_i \bx_i[n]$, and $F$  is the objective function of Problem  (\ref{Localization_Problem}). Note that $J[n]=0$ if and only if  $\bar{\bx}[n]$ is a stationary solution of  Problem (\ref{Localization_Problem}); therefore $J[n]$ measures the progress of the algorithms toward stationarity. The function $D[n]$ measures the disagreement among the agents' local variables; it converges to zero if an agreement is asymptotically achieved. Since the positions $\bx_0$ of the targets are known, we also report the normalized mean squared error (NMSE), evaluated at $\bar{\bx}[n]$:\vspace{-0.1cm}
\begin{equation}\label{NMSE}
{\rm NMSE}[n]=\frac{\|\bar{\bx}[n]-\bx_0\|^2}{\|\bx_0\|^2}. \vspace{-0.1cm}
\end{equation}
In Fig. \ref{Loc_comparison_a}   we plot $J[n]$ and $D[n]$  versus  the total
number of communication exchanges per node. For the D-Gradient, this number
coincides with the iteration index $n$; for NEXT, the number of communication
exchanges per node is $2\cdot n$. Fig.  \ref{Loc_comparison_b} shows ${\rm NMSE}[n]$  versus the number of exchanged messages.
All the curves are averaged over 500 independent noise realizations. In our experiments, the algorithms were observed to converge to the same consensual stationary solution. The  figures clearly show  that both versions of NEXT are much faster than the D-Gradient \cite{Bianchi-Jakubowicz} (or, equivalently, they require less information exchanges than the D-gradient), while having similar computational cost per iteration. This is mainly due to the fact that NEXT exploits the structure of the objective function (e.g., partial convexity, local information about the global objective). Moreover, as expected,  NEXT-PL reaches high precision  faster than NEXT-L; this is mainly due to the fact that the surrogate function (\ref{Loc_fi_tilde})  retains the partial convexity of functions  $f_{it}$ rather than just linearizing them.\vspace{-0.3cm}

\subsection{Distributed spectrum cartography in CR networks}

We consider now a convex instance of Problem (\ref{Problem}),   the estimation of the space/frequency power distribution in CR networks. This will permit to  compare NEXT with a variety of ad-hoc distributed schemes proposed for convex problems. We focus on the problem of building the spectrum cartography  of a given operational region, based on local measurements of the signal transmitted by the primary sources.
Utilizing a basis expansion model for the  power spectral density of the signal generated by the  sources, see e.g. \cite{Bazerque-Giannakis}, the spectrum cartography problem can be formulated as\vspace{-0.2cm}
\begin{equation}\label{CoopSensing_Problem}
\min_{\bx\in \mathcal{K}} \; \sum_{i=1}^I\|\boldsymbol{\varphi}_i-\bB_i\bx\|^2+\lambda\; \mathbf{1}^T\bx\vspace{-0.1cm}
\end{equation}
where $\bx\in \mathcal K$ collects all the powers emitted by the  primary sources over all the basis functions,
and $\mathcal{K}$ imposes constraints on the power emitted by the sources, e.g., non-negativity and upper bound limits; $\boldsymbol{\varphi}_i=\{\varphi_{ik}\}_{k=1}^{N_f}$ collects noisy samples of the received power over $N_f$ frequencies by sensor $i$; $\bB_i\in \mathbb{R}^{N_f\times N_bN_s}$ is the matrix of regressors, whose specific expression is not relevant in our discussion,  with $N_s$ and  {$N_b$} being the number of sources and the number of basis functions  \cite{Bazerque-Giannakis}, respectively; and $\lambda>0$ is a regularization parameter.

\indent Problem (\ref{CoopSensing_Problem}) is of course a convex instance of    (\ref{Problem}), with
$f_i(\bx)=\|\boldsymbol{\varphi}_i-\bB_i\bx\|^2$  and $G(\bx)=\lambda\;\mathbf{1}^T\bx$. Aiming at preserving the convexity of $f_i(\bx)$ as in (\ref{cvx_preserved_surrogate}), a natural choice for the surrogate functions is
\begin{equation}\label{Ui_coope_sens}
\widetilde{f}_i(\bx,\bx[n])=\|\boldsymbol{\varphi}_i-\bB_i\bx\|^2+\frac{\tau}{2}\|\bx-\bx[n]\|^2, \vspace{-0.1cm}
\end{equation}
with   $\tau$ being any positive constant.

\begin{figure}[t]
\centering
\includegraphics[width=7cm]{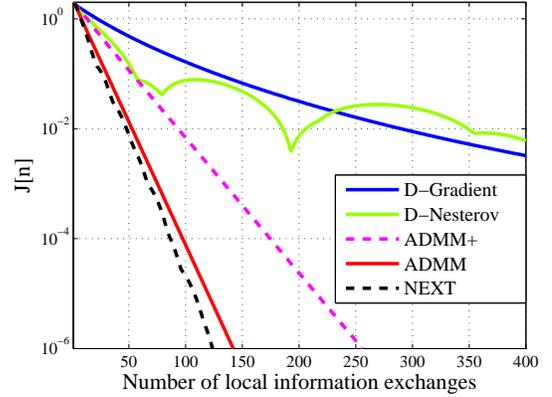}
\caption{Spectrum cartography  [cf. \eqref{CoopSensing_Problem}]: Distance from stationarity $J[n]$ versus number of local information  exchange.}\label{Coop_sens_comparison}\vspace{-0.4cm}
\end{figure}
\vspace{0.1cm}
\noindent{\textbf{Numerical Example.}} {We simulated a time-invariant connected network composed of  $I=30$ nodes randomly deployed over a 100 square meter area}. The setting is the following.  We consider $N_s=2$ active transmitters located at $[(2.5, 2.5); (7.5, 7.5)]$. The path-loss of each channel is assumed to follow the law  $g_{is}=1/(1+d_{is}^2)$, with $d_{is}$ denoting the distance between the $i$-th cognitive node and the $s$-th transmitter. The number of frequency basis functions is $N_b=10$, and each node scans $N_f=30$ frequency channels between 15 and 30 MHz. The first source transmits uniformly over the second, third, and fourth basis using a power budget equal to 1 Watt, whereas the second source transmits over the sixth, seventh, eighth, and ninth basis using a power equal to 0.5 Watt. The measures of the power spectral density at each node are affected by additive Gaussian zero-mean white noise, whose  variance is set so that the resulting  minimum SNR is equal to 3 dB. In (\ref{CoopSensing_Problem}), the regularization parameter is set to $\lambda=10^{-3}$, and the feasible set is $\mathcal{K}=[0,5]^{N_b\cdot N_s}$.\\ \indent {In our experiments, we compare five algorithms, namely: 1)   NEXT based on the surrogate function (\ref{Ui_coope_sens}); 2) the distributed gradient-based method in \cite{Bianchi-Jakubowicz} (termed as \emph{D-Gradient}); 3) the ADMM algorithm in the distributed form as proposed in  \cite{Bazerque-Giannakis}; 4) the distributed Nesterov-based method as developed in  \cite{Jakovetic-Xavier-Moura} (termed as \emph{D-Nesterov}); and 5) the distributed  ADMM+ algorithm \cite{Bianchi-Hachem-Iutzeler14}. We also tested the gradient scheme of \cite{Nedic-Ozdaglar-Parillo}, which performs as D-Gradient and thus is not reported. The free parameters in the above algorithms are set as follows.  NEXT exploits the local approximation  (\ref{Ui_coope_sens}) (with $\tau=0.8$), and uses the step-size rule  (\ref{step2}), with $\alpha[0]=0.1$ and $\mu=0.01$.  D-Gradient adopts the step-size rule in (\ref{step2}), with $\alpha[0]=0.5$ and $\mu=0.01$. Using the same notation as in  \cite{Bazerque-Giannakis}, in ADMM, the regularization parameter is set to $\alpha=0.015$. D-Nesterov  uses the step-size rule (\ref{step1}), with $\alpha_0=1.5$ and $\gamma=1$. Finally, using the same notation as  \cite{Bianchi-Hachem-Iutzeler14}, in ADMM+, we set $\tau=15.5$ and $\rho=2\tau$.} Our experiments showed that the above tuning leads to the fastest practical convergence speed for each algorithm. The weight coefficients in the consensus update of the algorithms (e.g., $w_{ij}[n]$ in (\ref{weights}) for  NEXT)  are chosen to be time-invariant and satisfying the Metropolis rule.

{In Figures \ref{Coop_sens_comparison}, \ref{D}, and \ref{Coop_sens_comparison2}, we plot $J[n]$, $D[n]$ [c.f. (\ref{Metric_J})], and NMSE$[n]$ [cf. (\ref{NMSE})]   versus the number of local information exchanges,  for the aforementioned five algorithms.  The results are averaged over 200 independent noise realizations. The analysis of the figures shows that, while all the algorithms converge to the globally optimal solution of (\ref{CoopSensing_Problem});  NEXT is significantly faster than the D-Gradient, D-Nesterov, and ADMM+, and exhibits performance similar to ADMM.}

\begin{figure}[t]
\centering
\includegraphics[width=7cm]{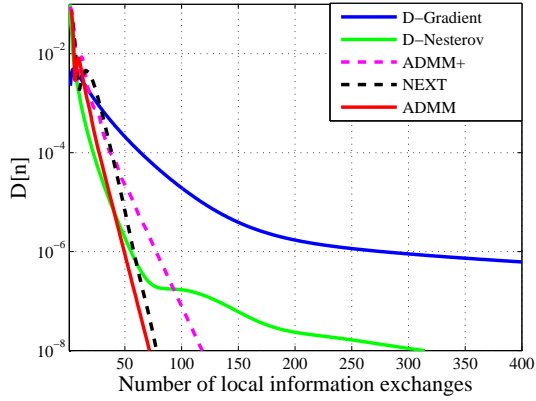}
\caption{Spectrum cartography  [cf. \eqref{CoopSensing_Problem}]: Consensus disagreement $D[n]$ versus number of local information exchanges.}\label{D}\vspace{-0.2cm}
\end{figure}
\begin{figure}[t]
\centering
\includegraphics[width=7cm]{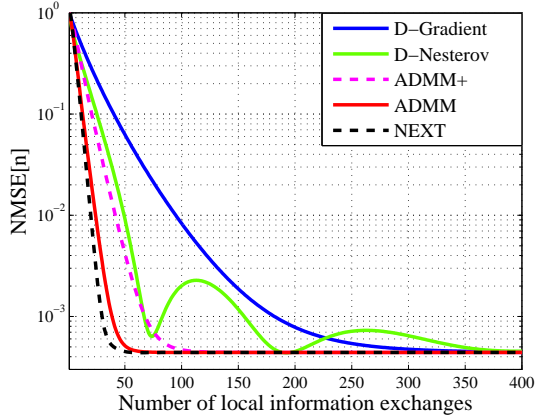}
\caption{Spectrum cartography  [cf. \eqref{CoopSensing_Problem}]: NMSE$[n]$ versus number of local information exchanges.}\label{Coop_sens_comparison2}\vspace{-0.4cm}
\end{figure}

\noindent{\textit{The effect of switching topology:}} One of the major features of the proposed approach is the capability to handle (nonconvex objectives and) time-varying network topologies, while  standard ADMM and D-Nesterov have no convergence guarantees (some asynchronous instances of ADMM have been recently proved to converge in a probability one sense \cite{Wei-Ozdaglar13}, \cite{iutzeler2013cdc}). To illustrate the robustness of NEXT to switching topologies, we run NEXT in the same setting as in Fig. \ref{Coop_sens_comparison} but simulating a $B$-strongly connected graph topology. We compare NEXT with the D-gradient  scheme  \cite{Bianchi-Jakubowicz}. In Fig. \ref{Topol_effect} we plot  $J[n]$ versus the number of local information exchanges, for different values of the  uniform graph connectivity coefficient  $B$ [cf. Assumption B1] (the case $B=1$ corresponds to a fixed strongly connected graph, and is reported as a benchmark). The results are averaged over 200 independent noise and graph realizations. The figure clearly shows that NEXT is much faster than the D-Gradient \cite{Bianchi-Jakubowicz}, for all values of the coefficient $B$. Furthermore, as expected, one can see that larger values of $B$ lead to a lower practical convergence speed, due to the slower diffusion of information over the network. As a final remark, we observe that NEXT requires the solution of a box-constrained least square problem at each iteration, due to the specific choice of the surrogate function in (\ref{Ui_coope_sens}). Even if this problem is computationally affordable, a closed form solution is not available. Thus, in this specific example, the improved practical convergence speed of NEXT is paid in terms of a higher cost per iteration than that of the D-gradient. Nevertheless, whenever  the (possibly inexact) computation of the solution of the convex subproblems in Step 2a) of Algorithm 1 (or Algorithm 2) is too costly, one can still reduce the computational burden by adopting a surrogate function as in (\ref{gradient_surrogate}), which linearize the function $f_i(\bx)=\|\boldsymbol{\varphi}_i-\bB_i\bx\|^2$ around the current point $\bx[n]$. In this case, the solution of the convex subproblems in Step 2a) of Algorithm 1 (or Algorithm 2) is given in closed form, thus making NEXT comparable to D-Gradient in terms of computational cost. \vspace{-0.4cm}
\begin{figure}[t]
\centering
  \includegraphics[width=6.9cm]{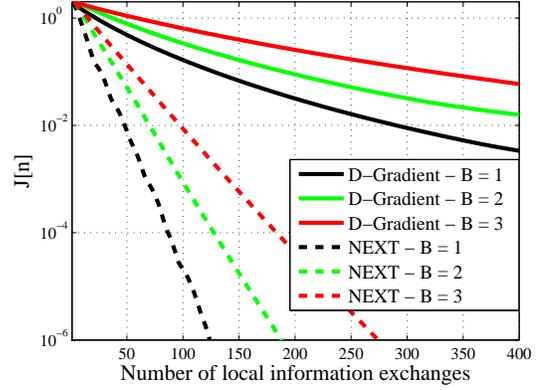}
  \caption{Spectrum cartography  [cf. \eqref{CoopSensing_Problem}]: Effect of the time-varying topology;
  $J[n]$ versus number of local information exchanges,  for different values of the uniform graph connectivity coefficient $B$.
  }\label{Topol_effect} \vspace{-0.4cm}
\end{figure}

\subsection{Flow Control in Communication Networks}

Consider  a network composed of a set  $\mathcal L\triangleq \{1,\ldots, L\}$ of unidirectional links of fixed capacities $c_l$, $l\in\mathcal  L$. The network is shared by a set $\mathcal I\triangleq \{1,\ldots, I\}$ of sources. Each source $i$  is  characterized by four parameters  $(\mathcal L_i,f_i,m_i,M_i)$. The path $\mathcal L_i\subseteq \mathcal L$ is the set of links used by the source $i$; $f_i:\mathbb{R}_+\rightarrow \mathbb{R}$ is the utility function, $m_i$ and $M_i$ are the minimum and maximum transmission rates, respectively. Source $i$ attains the utility $f_i(x_i)$ when it transmits at rate $x_i\in \mathcal{K}_i\triangleq [m_i, M_i]$. For each link $l\in \mathcal{L}$, let $\mathcal{S}_l\triangleq \{i\in \mathcal{I}\,:\,l\in
\mathcal{L}_i\}$ be the set of sources that use link $l$. The flow control problem can be then formulated as:
\begin{align}\label{flow_control}
& \underset{\bx\triangleq (x_i)_{i\in \mathcal I}}{\max}  \;\; \sum_{i\in  \mathcal I} f_i(x_i)  \\
& \,\,\,\text{\,\,\,s.\,t.}  \quad\quad \,  \sum_{i\in \mathcal{S}_l}x_i\leq c_l, \,\,  \forall l\in \mathcal L, \quad \text{and}\quad x_i\in \mathcal{K}_i,\,\,\forall i\in \mathcal I,\nonumber
\end{align}
where the first constraint imposes that the aggregate source rate at any link $l$ does not exceed the link capacity $c_l$. Depending on the kind of traffic (e.g., best effort or streaming traffic), the functions $f_i$ have different expressions. We consider
real-time applications, such as video streaming and voice over IP services,
which entail inelastic traffic. Unlike elastic data sources that are modeled by strictly concave utility functions \cite{Low-Lapsey}, inelastic sources are modeled by nonconcave utility functions \cite{Jafari}, in particular sigmoidal functions, defined as
\vspace{-0.1cm}
\begin{equation}\label{sigmoid}
f_i(x_i)=\frac{1}{1+e^{-(\alpha_i x_i+\beta_i)}},\vspace{-0.1cm}
\end{equation}
where  $\alpha_i>0$ and $\beta_i<0$ are integer.
\\ \indent Problem (\ref{flow_control}) is nonconvex;   standard (primal/)dual decomposition techniques widely used in the literature for convex instances (see, e.g., \cite{Low-Lapsey}) are no longer applicable, due  to a positive duality gap. The SCA framework in \cite{Scutari-Facchinei-Song-Palomar-Pang} for parallel nonconvex optimization can   conceivably be applied to   (\ref{flow_control}), but it requires   a fully connected network. Since (\ref{flow_control}) is a special case of Problem (\ref{Problem}), one can  readily apply our framework, as outlined next.  Observe preliminary that the sigmoid in (\ref{sigmoid}) can be rewritten as a DC  function  \vspace{-0.1cm}
\begin{equation}
f_i(x_i)=\underset{h(x_i)}{\underbrace{e^{(\alpha_i x_i+\beta_i)}}}-\underset{g(x_i)}{\underbrace{\frac{e^{2(\alpha_i x_i+\beta_i)}}{1+e^{(\alpha_i x_i+\beta_i)}}}},\nonumber\vspace{-0.2cm}
\end{equation}
where $h(x_i)$ and $g(x_i)$ are convex functions.
{Let $\bx_i[n]\triangleq (x_{i,j}[n])_{j=1}^I$ be the local copy of the common vector $\bx=(x_j)_{j=1}^{I}$, made by user $i$ at iteration $n$; $x_{i,i}[n]$ is the transmission rate of  source $i$. Then,  a natural choice for the surrogate $\widetilde{f}_i(\bx_{i};\bx_i[n])$ is retaining the concave part $-g$ while linearizing the convex part $h$, i.e., $\widetilde{f}_i(\bx_i;\bx_i[n])=-g_i(x_{i,i})+h'(x_{i,i}[n]) (x_{i,i}-x_{i,i}[n])+\frac{\tau_i}{2}\|\bx_{i}-\bx_{i}[n]\|^2$.}\vspace{-0.1cm}

\subsection{Sparse Distributed Maximum Likelihood Estimation}

Consider a distributed estimation problem in wireless sensor networks, where $\bx\in \mathcal{K}$ denotes the parameter vector to be estimated, modeled as the outcome of a random variable, described by a known prior probability density function (pdf) $p_{X}(\bx)$. Let us denote by $\boldsymbol{\varphi}_i$ the measurement vector collected by node $i$, for $i=1,\ldots,I$. Assuming statistically independent observations and a Laplacian prior on the variable $\bx$, the maximum likelihood estimate can be obtained as the solution of the following problem (see, e.g., \cite{Barb-Sard-Dilo}):\vspace{-0.2cm}
\begin{equation}\label{est_problem}
\max_{\bx\in \mathcal{K}} \;\; \sum_{i=1}^{I} \;\log p_{\Phi/X}(\boldsymbol{\varphi}_i/\bx)+\lambda\|\bx\|_1\vspace{-0.2cm}
\end{equation}
where $p_{\Phi/X}(\boldsymbol{\varphi}_i/\bx)$ is the pdf of $\boldsymbol{\varphi}_i$ conditioned to $\bx$, and $\lambda>0$ is a regularization parameter. Problem (\ref{est_problem}) is clearly an instance of our general formulation (\ref{Problem}), with each $f_i(\bx)=\;\log p_{\Phi/X}(\boldsymbol{\varphi}_i/\bx)$ and
$G(\bx)=\lambda\|\bx\|_1$. Note that, differently from classical  formulations of   distributed estimation problems (e.g., \cite{Barb-Sard-Dilo}, \cite{Schizas-Giann-Roum-Rib}), here we do not assume that the pdf's $p_{\Phi/X}(\boldsymbol{\varphi}/\bx)$ in (\ref{est_problem}) are log-concave functions of  $\bx$, which leads to generally nonconvex optimization problems. To the best of our knowledge, NEXT is the first distributed scheme converging to (stationary) solutions of \emph{nonconvex} estimation problems in the general form (\ref{est_problem}). Potential applications of  (\ref{est_problem}) abound in the contexts of environmental and civil infrastructures monitoring, smart grids, vehicular networks, and CR networks, just to name a few.\vspace{-0.4cm}

\subsection{Other Applications}
The proposed algorithmic framework can be  applied to a variety of other applications, including i) network localization \cite{Simonetto-Leus}; ii) multiple vehicle coordination \cite{Raffard-Tomlin-Boyd}; and iii) resource allocation over vector Gaussian Interference channels \cite{Scutari-Facchinei-Song-Palomar-Pang}. Quite interestingly, convergence of our algorithm is guaranteed under weaker conditions than those assumed in the aforementioned papers (e.g., time-varying topology, not fully connected networks). We omit the details because of the space limitation.\vspace{-0.2cm}

\section{Conclusions}

In this paper we have introduced NEXT, a novel (best-response-based) algorithmic framework for nonconvex distributed optimization in multi-agent networks with time-varying (nonsymmetric) topology. NEXT exploits successive convex approximation techniques while leveraging dynamic consensus as a mechanism to distribute the computation as well as propagate the needed information over the network.
We finally customized NEXT to a variety of (convex and nonconvex) optimization problems in  several fields, including  signal processing, communications, networking, machine learning, and cooperative control. Numerical results  show that the new method compares favorably to existing distributed algorithms on both convex and nonconvex  problems.\vspace{-0.2cm}
\section{Acknowledgments}
The authors are grateful to Francisco Facchinei for his invaluable comments.\vspace{-0.2cm}
\appendix
We preliminarily introduce some definitions and intermediate results  that are instrumental to prove Theorems \ref{simplified_convergence_th} and \ref{main_th}. We proceed considering the evolution of Algorithm 2. Similar results apply to Algorithm 1 as a special case.

\vspace{-0.2cm}
\subsection{Notation and preliminary results}\label{Appendix_A}
\noindent \textbf{Averages sequences}: Given the sequences generated by Algorithm 2,  let us introduce the following column vectors: for every $n\geq 0$ and $l>n$,
\begin{align}\label{delta_r}
\bx[n]&=\left [ \bx_1[n]^T,\ldots ,\bx_I[n]^T\right]^T;\medskip\nonumber\\
\by[n]&=\left [ \by_1[n]^T,\ldots ,\by_I[n]^T\right]^T; \medskip\nonumber\\
\br[n]&= \left [ \nabla f_1[n]^T,\ldots ,\nabla f_I[n]^T\right]^T;\medskip\\
\mathbf{\Delta r}_i[l,n]&= \nabla f_i[l]-\nabla f_i[n]; \medskip\nonumber\\
\mathbf{\Delta r}[l,n]&= \left[\mathbf{\Delta r}_1[l,n]^T,\ldots, \mathbf{\Delta r}_I[l,n]^T\right]^T;\nonumber
\end{align}
and the associated ``averages" across the users
\begin{align}\label{x_y_bar}
\mathbf{\overline{x}}[n]=\frac{1}{I}\sum_{i=1}^I \bx_i[n],\quad\mathbf{\overline{y}}[n]=\frac{1}{I}\sum_{i=1}^I   \by_i[n],
\end{align}
\begin{align}
\overline{\mathbf{r}}[n]&=\frac{1}{I}\sum_{i=1}^I \nabla f_i[n],\,\,\,\,
\mathbf{\overline{\Delta r}}[l,n]=\frac{1}{I}\sum_{i=1}^I \mathbf{\Delta r}_i[l,n]. \label{delta_r_bar}
\end{align}
Note that  $\by[0]=\br[0]$, and
\begin{align}\label{r_bar_decomp}
\mathbf{\overline{r}}[n]-\mathbf{\overline{r}}[0]=\sum_{l=1}^n\mathbf{\overline{\Delta r}}[l,l-1].
\end{align}

The proof of convergence of Algorithm 2 relies on studying the so-called ``consensus error dynamics'', i.e.  the distance of the  agents' variables $\x_i[n]$ from their average consensus dynamic $\mathbf{\overline{x}}[n]$.  To do so, we need to introduce some auxiliary sequences, as given next. Let us define, for $n\geq 0$ and $l>n$,
\begin{align}
\nabla f_i^{\rm av}[n]&=\nabla f_i(\bar{\bx}[n]),\bigskip\nonumber\\
\mathbf{\Delta r}^{\rm av}_i[l,n]&= \nabla f^{\rm av}_i[l]-\nabla f^{\rm av}_i[n], \bigskip\nonumber\\
\widetilde{\bx}^{\rm av}_{i}[n]&\triangleq \underset{\bx_i\in\mathcal{K}}{\text{argmin}} \,\,\widetilde{U}_i\left(\bx_i;\bar{\bx}[n],\widetilde{\boldsymbol{\pi}}^{\rm av}_{i}[n]\right) \bigskip\label{aux_averages}\\
\by_i^{\rm av}[n+1]&=\sum_{j\in\mathcal{N}_i^{\rm in}[n]} w_{ij}[n]\by_j^{\rm av}[n]+\mathbf{\Delta r}^{\rm av}_i[n+1,n]\bigskip\nonumber\\
\widetilde{\boldsymbol{\pi}}_{i}^{\rm av}[n]&=I\by^{\rm av}_i[n]-\nabla f_{i}^{\rm av}[n], \nonumber
\end{align}
with $\by_i^{\rm av}[0]\triangleq \nabla f_i^{\rm av}[0]$; and their associated vectors (and averages across the agents)\vspace{-0.1cm}
\begin{align}\label{aux_averages_vec}
\mathbf{\bar{r}}^{\rm av}[n]&=\frac{1}{I}\sum_{i=1}^I \nabla f^{\rm av}_i[n],\bigskip\nonumber\\
\br^{\rm av}[n]&= \left [ \nabla f^{\rm av}_1[n]^T,\ldots ,\nabla f^{\rm av}_I[n]^T\right]^T,\medskip\\
\mathbf{\Delta r}^{\rm av}[l,n]&= \left[\mathbf{\Delta r}^{\rm av}_1[l,n]^T,\ldots, \mathbf{\Delta r}^{\rm av}_I[l,n]^T\right]^T,\medskip\nonumber\\
\by^{\rm av}[n]&=\left [ \by^{\rm av}_1[n]^T,\ldots ,\by^{\rm av}_I[n]^T\right]^T. \nonumber
\end{align}
\noindent Note that $\widetilde{\bx}^{\rm av}_{i}[n]$ can be interpreted as the counterpart of $\widetilde{\bx}_{i}[n]$, when in (S.2a), $\x_i[n]$ and $\widetilde{\boldsymbol{\pi}}_{i}[n]$ are replaced by $\overline{\x}[n]$ and $\widetilde{\boldsymbol{\pi}}_{i}^{\rm av}[n]$, respectively; and so do  $\by_i^{\rm av}[n]$  and $\widetilde{\boldsymbol{\pi}}_{i}^{\rm av}[n]$. In Appendix \ref{Prop_consensus}, we show that the consensus error dynamics $\|\x_i[n]-\overline{\x}[n]\|$ as well as the best-response deviations $\|\widetilde{\bx}^{\rm av}_{i}[n]-\widehat{\bx}_{i}(\overline{\x}[n])\|$ and $\|\widetilde{\bx}_{i}[n]- \widetilde{\bx}^{\rm av}_{i}[n]\|$ asymptotically vanish, which is a key result to prove Theorems 3 and 4.

\noindent \textbf{Properties of $\widehat{\bx}_{i}(\bullet)$}: The next proposition introduces some key properties of the best-response map (\ref{best_resp_x_hat_2}), which will be instrumental to prove convergence; the proof of the proposition follows similar arguments as \cite[Prop. 8]{Scutari-Facchinei-Sagratella}, and thus is omitted.

\vspace{-.2cm}
\begin{proposition}\label{Prop_best_response} Under Assumption A and F1-F3,  each best response map $\widehat{\bx}_{i}(\bullet)$ in (\ref{best_resp_x_hat_2}) enjoys  the following properties:
\begin{description}
  \item[(a)] $\widehat{\bx}_{i}(\bullet)$ is Lipschitz continuous on $\mathcal{K}$, i.e., there exists a positive constant $\widehat{L}_i$ such that
      \begin{align}\label{p1a}
         \|\widehat{\bx}_{i}(\bw)-\widehat{\bx}_{i}(\bz)\|\leq \widehat{L}_i\;\|\bw-\bz\|,\quad \forall\;\bw,\bz\in \mathcal{K};
      \end{align}
  \item[(b)] The set of the fixed-points of  $\widehat{\bx}_{i}(\bullet)$ coincides with the set of stationary  solutions of Problem \eqref{Problem}; therefore $\widehat{\bx}_{i}(\bz)$ has a fixed-point;
  \item[(c)] For every given $\bz\in{\cal K}$,
      \begin{align}\label{p1c}
        & (\widehat{\bx}_{i}(\bz)-\bz)^T \nabla F(\bz) + G\left(\widehat{\bx}_{i}(\bz)\right) - G(\bz)\smallskip\nonumber\\
        &\qquad\qquad \leq - c_{\tau} \|\widehat{\bx}_{i}(\bz)-\bz\|^2,
      \end{align}
      with $c_{\tau}\triangleq \min_i \tau_i>0$.
  {\item[(d)] There exists a finite constant $\eta>0$ such that
  \begin{align}\label{p1d}
         \|\widehat{\bx}_{i}(\bz)-\bz\|\leq\eta, \quad \forall\; \bz\in \mathcal{K}.
      \end{align}}
\end{description}
\end{proposition}

\noindent \textbf{Transition matrices  and their properties}: For any $n$ and $l\leq n$,  let us introduce the transition matrix $\mathbf{P}[n,l]$, defined as
\begin{equation}\label{P_left}
\mathbf{P}[n,l]\triangleq\bW[n]\bW[n-1]\dots\bW[l], \quad n\geq l,
\end{equation}
with $\mathbf{P}[n,n]=\bW[n]$ for all $n$, and  $\bW[n]$ given in (\ref{weights}).  Let
\begin{align}\label{w_hat}
&\widehat{\bW}[n]\triangleq \bW[n]\,\otimes\,\bI_m, \medskip\nonumber\\
&\widehat{\mathbf{P}}[n,l]\triangleq \widehat{\bW}[n]\widehat{\bW}[n-1]\cdots\widehat{\bW}[l],\quad n\geq l,
\end{align}
with $\widehat{\mathbf{P}}[n,n]=\widehat{\bW}[n]$, for all $n$. The above matrices will be used to  describe the evolution of the agent variables associated with Algorithm 2. Introducing
\begin{equation}\label{def_J}
  \mathbf{J}=\frac{1}{I}\mathbf{1}\mathbf{1}^T\,\otimes\,\bI_m, \quad \hbox{and}\quad \mathbf{J}_\perp=\bI_{m\,I}-\bJ,
\end{equation}
it is not difficult to check that the following holds:
\begin{equation}\label{Jproperties}
\begin{array}{cc}
\bJ_\perp\widehat{\bW}[n]=\bJ_\perp\widehat{\bW}[n]\bJ_\perp=\widehat{\bW}[n]-\left(\frac{1}{I}\mathbf{1}\mathbf{1}^T\otimes\bI_m\right),\medskip\\
\widehat{\mathbf{P}}[n,l]=\mathbf{P}[n,l]\otimes \bI_m,
\end{array}
\end{equation}
and
\begin{equation}\label{P_hat_properties}
\begin{array}{ll}
&\bJ_\perp\widehat{\bW}[n]\,\bJ_\perp\widehat{\bW}[n-1]\,\cdots\,\bJ_\perp\widehat{\bW}[l]=\bJ_\perp\widehat{\mathbf{P}}[n,l] \medskip\\
&\qquad= \left(\mathbf{P}[n,l]-\dfrac{1}{I}\mathbf{1}_I\mathbf{1}_I^T\right)\otimes\mathbf{I}_m, \quad \forall\; n\geq l.
\end{array}
\end{equation}
The next lemma shows that, under the topology assumptions B1-B2,  the matrix difference $\mathbf{P}[n,l]-\frac{1}{I}\mathbf{11}^T$ decays geometrically, thus enabling consensus among agent variables.

\vspace{-.2cm}
\begin{lemma}\label{Lemma_left_product}
Given $\mathbf{P}[n,l]$ in (\ref{P_left}), under B1-B2,  the following holds:
\begin{equation}\label{norm_P}
\left\|\mathbf{P}[n,l]-\dfrac{1}{I}\mathbf{11}^T\right\|\leq c_0\, \rho^{n-l+1},\quad \forall\; n\geq l,
\end{equation}
for some $c_0>0$, and $\rho\in(0,1)$.
\end{lemma}
\begin{proof}
The proof is based on \cite[Lemma 5.2.1]{Tsitsiklis}. %in case of double stochastic matrices.
\end{proof}\smallskip

\vspace{-.2cm}
\noindent\textbf{Some useful Lemmas:} We introduce two lemmas dealing with convergence of some sequences.

\vspace{-.2cm}
\begin{lemma}\label{Lemma_sequences}
Let $0<\lambda<1$, and  let $\{\beta[n]\}$ and  $\{\nu[n]\}$ be two positive scalar sequences. Then, the following   hold:
\begin{description}
  \item[(a)] If $\displaystyle\lim_{n\rightarrow\infty}\beta[n]=0$, then\vspace{-0.2cm}
   \begin{equation}\label{lemma1a}
       \lim_{n\rightarrow\infty}\,\sum_{l=1}^{n}\lambda^{n-l}\beta[l]=0.
    \end{equation}
  \item[(b)] If $\sum_{n=1}^\infty \beta[n]^2<\infty$ and $\sum_{n=1}^\infty \nu[n]^2<\infty$, then
   \begin{align}
     & \hbox{\rm (b.1):}\;\;\displaystyle\lim_{n\rightarrow\infty}\, \sum_{k=1}^{n}\sum_{l=1}^{k}\lambda^{k-l}\beta[l]^2<\infty, \label{lemma1b1}\\
     & \hbox{\rm (b.2):}\;\;\displaystyle\lim_{n\rightarrow\infty}\, \sum_{k=1}^{n}\sum_{l=1}^{k}\lambda^{k-l}\beta[k]\nu[l]<\infty. \label{lemma1b2}
   \end{align}
\end{description}
\end{lemma}

\begin{proof} The lemma is a refinement of similar results in \cite{Nedic-Ozdaglar-Parillo}.

\noindent (a) and (b.1) can be found in \cite[p. 931]{Nedic-Ozdaglar-Parillo};  (b.2) is proved next. The following bound holds: \vspace{-0.2cm}
\begin{align}\label{lemma1proof}
\sum_{k=1}^{n}\sum_{l=1}^{k}&\lambda^{k-l}\beta[k]\nu[l] \nonumber\\
&  \hspace{-.5cm}{\leq}\,\, \frac{1}{2}\,\sum_{k=1}^{n}\beta[k]^2\sum_{l=1}^{k}\lambda^{k-l} +\frac{1}{2}\,\sum_{k=1}^{n}\sum_{l=1}^{k}\lambda^{k-l} \nu[l]^2 \nonumber\\
& \hspace{-.5cm} {\leq} \,\,\frac{1}{2(1-\lambda)}\,\sum_{k=1}^{n}\beta[k]^2+\frac{1}{2}\,\sum_{k=1}^{n}\sum_{l=1}^{k}\lambda^{k-l} \nu[l]^2,
\end{align}
where we used the inequalities $\beta[k]\gamma[l]\leq (\beta[k]^2+\nu[l]^2)/2$ and $\sum_{l=1}^{k}\lambda^{k-l}\leq \frac{1}{1-\lambda}$, for all $k$. The proof of (\ref{lemma1b2}) follows  from (\ref{lemma1proof}) invoking    $\sum_{n=0}^\infty \beta[n]^2<\infty$ and   (\ref{lemma1b1}).
\end{proof}

\vspace{-.3cm}
\begin{lemma} [{\cite[Lemma 1]{Ber_Tsitsi}}]\label{lemma_Robbinson_Siegmunt}
Let $\{Y[n]\}$, $\{X[n]\}$, and $\{Z[n]\}$ be three sequences
of numbers such that $X[n]\geq0$ for all $n$. Suppose that
\vspace{-.1cm}
\[
Y[n+1]\leq Y[n]-X[n]+Z[n],\quad\hbox{for all $n$},
\]
and $\sum_{n=1}^\infty Z[n]<\infty$. Then, either $Y[n]\rightarrow-\infty$
or else $\{Y[n]\}$ converges to a finite value and $\sum_{n=1}^\infty X[n]<\infty$.
\end{lemma}

\vspace{-.1cm}
\noindent\textbf{On the iterates $\bx[n]$ and $\bar{\bx}[n]$:} We conclude this section  writing in a compact form the evolution of the sequences $\bx[n]$ and $\bar{\bx}[n]$ generated by  Algorithm 2; this  will be largely used in forthcoming proofs. Combining S.2(b) and S.3(a) of Algorithm 2,   $\bx[n]$, defined in  (\ref{delta_r}), can be written as:
\vspace{-.1cm}
\begin{equation}\label{x_evo}
\bx[n]=\widehat{\bW}[n-1]\bx[n-1]+\alpha[n-1]\,\widehat{\bW}[n-1]\Delta\bx^{\texttt{inx}}[n-1]\vspace{-0.1cm}
\end{equation}
where $\Delta\bx^{\texttt{inx}}[n]\triangleq (\bx^\texttt{inx}_i[n]-\bx_i[n])_{i=1}^I$. Using  (\ref{x_evo}), the evolution of the average vector $\bar{\bx}[n]$, defined in (\ref{x_y_bar}), reads
\begin{align} \label{x_bar_evo}
\bar{\bx}[n]%&=\dfrac{1}{I}  \left(\mathbf{1}_I^T\otimes \bI_m\right)\bx[n]\nonumber\smallskip\\
%&\hspace{-.4cm}
=\bar{\bx}[n-1]+\,\dfrac{\alpha[n-1]}{I}\left(\mathbf{1}_I^T\otimes \bI_m\right)\Delta\bx^{\texttt{inx}}[n-1].
\end{align}

\vspace{-.6cm}
\subsection{Consensus Achievement}\label{Prop_consensus}

The following proposition establishes a connection between the sequence $(\bx_i[n])_{i=1}^I$ generated by Algorithm 2 and the sequences $\bar{\bx}[n]$, $\widetilde{\bx}^{\rm av}_i[n]$, and $\widehat{\bx}_i(\bar{\bx}[n])$, introduced in Appendix A. It shows that i) the agents reach a consensus on the estimates $\bx_i[n]$, $\forall i$, as $n\rightarrow\infty$; and ii) the limiting behavior of the best-responses $\widetilde{\bx}_i[n]$ is the same as $\widetilde{\bx}^{\rm av}_i[n]$ and $\widehat{\bx}_i(\bar{\bx}[n])$. This is a key result to prove convergence of Algorithm 2.

\vspace{-.1cm}
\begin{proposition}\label{Prop_cons} Let $\{(\bx_i[n])_{i=1}^I\}_n$ be the sequence generated by Algorithm 2, in the setting of Theorem \ref{main_th}. Then, for all $n\geq 0$, the following holds:
\begin{description}
  \item[(a)] {\rm [Bounded disagreement]:}
  \begin{eqnarray}\label{P2a}
     \left\|\bx_i^{\texttt{inx}}[n]-\bx_i[n]\right\|\leq c,
  \end{eqnarray}
  for all $i=1,\ldots,I$, and some finite constant $c>0$;
  \item[(b)] {\rm [Asymptotic agreement on $\bx_i[n]$]:}
     \begin{align}
        & \lim_{n\rightarrow\infty}\left\|\bx_i[n]-\bar{\bx}[n]\right\|=0, \label{P2b1}\\
        & \sum_{n=1}^\infty \alpha[n] \left\|\bx_i[n]-\bar{\bx}[n]\right\|<\infty, \label{P2b2}\\
        &  \sum_{n=1}^\infty \left\|\bx_i[n]-\bar{\bx}[n]\right\|^2<\infty. \label{P2b3}
     \end{align}
     for all $i=1,\ldots,I$;
  \item[(c)]  {\rm [Asymptotically vanishing tracking error]:}
     \begin{align}
        & \lim_{n\rightarrow\infty}\left\|\widetilde{\bx}^{\rm av}_i[n]-\widehat{\bx}_i(\bar{\bx}[n])\right\|=0 \label{P2c1}\\
        & \sum_{n=1}^\infty \alpha[n]\left\|\widetilde{\bx}^{\rm av}_i[n]-\widehat{\bx}_i(\bar{\bx}[n])\right\|<\infty, \label{P2c2}
     \end{align}
     for all $i=1,\ldots,I$.
  \item[(d)] {\rm [Asymptotic agreement on best-responses]:}
     \begin{align}
        &\lim_{n\rightarrow\infty}\left\|\widetilde{\bx}_i[n]-\widetilde{\bx}^{\rm av}_i[n]\right\|=0 \label{P2d1}\\
        &\sum_{n=1}^\infty \alpha[n] \left\|\widetilde{\bx}_i[n]-\widetilde{\bx}^{\rm av}_i[n]\right\|<\infty, \label{P2d2}
     \end{align}
     for all $i=1,\ldots,I$.

\end{description}
\end{proposition}

\begin{proof}  \textbf{Point (a):} The minimum principle of (\ref{best_resp_x_hat_3}) leads to
\begin{equation}
\left(\bx_i[n]-\widetilde{\bx}_i[n]\right)^T\Big(\nabla \widetilde{f}_i(\widetilde{\bx}_i[n];\bx_i[n])+\widetilde{\boldsymbol{\pi}}_{i}[n]+\partial G(\widetilde{\bx}_i[n])\Big)\geq 0.\label{min_princ}
\end{equation}
Using S.3(c) and F2,  (\ref{min_princ}) becomes
\begin{align}
&(\bx_i[n]-\widetilde{\bx}_i[n])^T\Big(\nabla \widetilde{f}_i(\widetilde{\bx}_i[n];\bx_i[n])-\nabla \widetilde{f}_i(\bx_{i}[n];\bx_i[n])\Big) \nonumber\\
&+ (\bx_i[n]-\widetilde{\bx}_i[n])^T \left(I\cdot \by_i[n]+ \partial G(\widetilde{\bx}_i[n])\right)\geq 0, \nonumber
\end{align}
which, invoking F1 [i.e., the uniformly strong convexity of $\widetilde{f}_i(\bullet;\bx_i[n])$], and $\|\partial G(\widetilde{\bx}_i[n])\|\leq L_G$  (cf. A5), yields
\begin{align}\label{min_princ3}
&\left\|\widetilde{\bx}_i[n]-\bx_i[n]\right\|  \leq \frac{I}{\tau_i} \left\|\by_i[n]\right\|+\frac{L_G}{\tau_i} \leq \,\frac{I}{\tau_i}\,\left\|\by[n]\right\|+\frac{L_G}{\tau_i}\smallskip\nonumber\\
&\quad   {\leq}\; \frac{I}{\tau_i} \, \|\by[n]-\mathbf{1}_I\otimes\mathbf{\overline{r}}[n]\| +  \frac{I}{\tau_i} \, \|\mathbf{1}_I\otimes\mathbf{\overline{r}}[n]\|+\frac{L_G}{\tau_i}\smallskip\nonumber\\
& \quad {\leq}\; \frac{I}{\tau_i} \,\|\by[n]-\mathbf{1}_I\otimes\mathbf{\overline{r}}[n]\| +c_1
\end{align}
where the last inequality follows from  A4, for  some finite $c_1\geq \frac{I}{\tau_i} \,\|\mathbf{1}_I\otimes\mathbf{\overline{r}}[n]\|+\frac{L_G}{\tau_i}$. From (\ref{min_princ3}), exploiting (\ref{precision}) and the triangle inequality, we then obtain
\begin{align}
&\left\|\bx^{\texttt{inx}}_i[n]-\bx_i[n]\right\| \leq  c_1+\varepsilon_i[n]+\frac{I}{\tau_i} \,\|\by[n]-\mathbf{1}_I\otimes\mathbf{\overline{r}}[n]\|. \nonumber
\end{align}
Since  $\varepsilon_i[n]$ is uniformly bounded [cf. (\ref{errors})], to complete the proof it is sufficient to  show  that $\|\by[n]-\mathbf{1}_I\otimes\mathbf{\overline{r}}[n]\|$ is bounded.
Note that  $\by[n]$ in S.3(b) and $\mathbf{1}_I\otimes\mathbf{\overline{r}}[n]$ can be rewritten as
\begin{align}\label{y_evol}
\by[n]&=\widehat{\bW}[n-1]\by[n-1]+\mathbf{\Delta r}[n,n-1]\nonumber\\
&\hspace{-.75cm}=\widehat{\mathbf{P}}[n-1,0]\br[0]+\sum_{l=1}^{n-1} \widehat{\mathbf{P}}[n-1,l] \mathbf{\Delta r}[l,l-1] +\mathbf{\Delta r}[n,n-1]
\end{align}
\vspace{-.5cm}
\begin{equation}\label{Kron_r_bar}
\hbox{\rm and} \qquad\mathbf{1}_I\otimes\mathbf{\overline{r}}[n]=\mathbf{J}\,\mathbf{{r}}[0]+\sum_{l=1}^n \mathbf{J}\,{\mathbf{\Delta r}}[l,l-1],
\end{equation}
\noindent respectively, where (\ref{y_evol}) follows from  (\ref{w_hat}), and in (\ref{Kron_r_bar}) we used
(\ref{r_bar_decomp}) and (\ref{def_J}). Then, the following holds:
\vspace{-.1cm}
\begin{align}\label{y_evol2}
& \by[n]-\mathbf{1}_I\otimes\mathbf{\overline{r}}[n]\nonumber\\
 &\hspace{-.2cm}\stackrel{(a)}{=}\left(\widehat{\mathbf{P}}[n-1,0]-\bJ\right)\br[0]+\sum_{l=1}^{n-1}  \left(\widehat{\mathbf{P}}[n-1,l]-\bJ\right) \mathbf{\Delta r}[l,l-1] \nonumber\\
 &\quad\quad+ \bJ_\perp\mathbf{\Delta r}[n,n-1]\nonumber\\
  &\hspace{-.2cm} \stackrel{(b)}{=} \left[\left(\mathbf{P}[n-1,0]-\dfrac{1}{I}\mathbf{1}_I\mathbf{1}_I^T\right)\otimes\mathbf{I}_m\right]\br[0]
  +\bJ_\perp\mathbf{\Delta r}[n,n-1]\nonumber\\
 & \quad+\sum_{l=1}^{n-1}  \left[\left(\mathbf{P}[n-1,l]-\dfrac{1}{I}\mathbf{1}_I\mathbf{1}_I^T\right)\otimes\mathbf{I}_m \right]\mathbf{\Delta r}[l,l-1]
\end{align}
where (a) follows from (\ref{y_evol}), (\ref{Kron_r_bar}), and (\ref{def_J}); and (b) is due to (\ref{Jproperties}) and (\ref{P_hat_properties}).
It follows from Lemma \ref{Lemma_left_product} that
$$\left\|\left(\mathbf{P}[n,l]-(1/I)\mathbf{1}_I\mathbf{1}_I^T\right)\otimes\mathbf{I}_m \right\|\leq c_0\rho ^{n-l+1},$$
with $\rho<1$, which together with A4 yields \vspace{-.2cm}
\begin{align}\label{y_evol3}
 \left\|\by[n]-\mathbf{1}_I\otimes\mathbf{\overline{r}}[n] \right\|&\leq \rho^{n} c_2 + c_3\sum_{l=1}^{n-1} \rho^{n-l}+c_4\nonumber\\
&\leq \rho\, c_2+\frac{c_3}{1-\rho}+c_4<+\infty,
\end{align}
for some positive, finite constants  $c_2$, $c_3$, $c_4$.% Using (\ref{y_evol3}) in (\ref{min_princ3}), the proof of point (a) follows.

\noindent \textbf{Point (b):}  We prove next (\ref{P2b1}). Note that, using (\ref{def_J}), we get
\begin{equation}\label{delta_x}
\bx[n]- \mathbf{1}_I\otimes\bar{\bx}[n]=\bx[n]-  \bJ {\bx}[n]=\bJ_\perp   \bx[n].
\end{equation}  Therefore, denoting by $\bx_\perp [n]\triangleq \bJ_\perp \bx[n]$, it is sufficient to show that
$\lim_{n\rightarrow\infty}\|\bx_\perp [n]\|=0.$
Exploiting (\ref{x_evo}), one can then write
\begin{align}\label{x_perp_evo}
\hspace{-0.25cm}\bx_\perp[n]&\overset{(a)}{=}\bJ_\perp\widehat{\bW}[n-1]\bx_\perp[n-1] \nonumber\\
&\quad+\alpha[n-1]\bJ_\perp\widehat{\bW}[n-1]\Delta\bx^{\texttt{inx}}[n-1] \nonumber\smallskip\\
&\hspace{-0.5cm}\overset{(b)}{=} \left[\left(\mathbf{P}[n-1,0]-\dfrac{1}{I}\mathbf{1}_I\mathbf{1}_I^T\right)\otimes\mathbf{I}_m\right]\bx_\perp[0]\smallskip\nonumber\\
 &\hspace{-0.6cm}+\sum_{l=0}^{n-1}  \left[\left(\mathbf{P}[n-1,l]-\dfrac{1}{I}\mathbf{1}_I\mathbf{1}_I^T\right)\otimes\mathbf{I}_m \right]\alpha[l]\,\Delta\bx^{\texttt{inx}}[l]
\end{align}
where (a) follows from (\ref{x_evo}) and (\ref{Jproperties}); and (b) is due to  (\ref{P_hat_properties}). Now, since  $\|\Delta\bx^{\texttt{inx}}[n]\|$ is bounded [cf. (\ref{P2a})] and $\|\bx_\perp[0]\|<\infty$, invoking Lemma \ref{Lemma_left_product},   (\ref{x_perp_evo}) yields
\begin{align}\label{x_ort_norm}
\|\bx_\perp[n]\|\leq c_5\,\rho^{n}+c_6\,\rho\, \sum_{l=1}^{n}\rho^{n-l}\alpha[l-1] \underset{n\rightarrow\infty}{\longrightarrow} 0,
\end{align}
for some positive finite constants $c_5$, $c_6$, where the last implication follows from (\ref{step-size}) and Lemma \ref{Lemma_sequences}(a).

We prove now (\ref{P2b2}). Using (\ref{delta_x}) one can write:
\begin{align}
&\lim_{n\rightarrow \infty}\sum_{k=1}^{n}\alpha[k]\left \|\bx_i[k]-\bar{\bx}[k]\right \| \leq \lim_{n\rightarrow \infty}\sum_{k=1}^{n}\alpha[k]\|\bx_\perp[k]\|\nonumber\smallskip\\
& \overset{(a)}{\leq} \lim_{n\rightarrow \infty}  \left [c_5 \sum_{k=1}^{n}\rho^{k}\alpha[k] +c_6\,\rho \sum_{k=1}^{n}\sum_{l=1}^{k}\rho^{k-l}\alpha[k]\alpha[l-1]\right]\overset{(b)}{<} \infty, \nonumber
\end{align}
for all $i=1,\ldots,I$, where (a) follows from (\ref{x_ort_norm}); and (b) is due to (\ref{step-size}) and Lemma \ref{Lemma_sequences}(b.2).

Finally, we prove (\ref{P2b3}). From (\ref{x_ort_norm}), we get
\begin{align}\label{x_ort_norm2_sum}
&\lim_{n\rightarrow\infty}\sum_{k=1}^{n}\|\bx_\perp[k]\|^2\nonumber\\
&\quad\leq \lim_{n\rightarrow\infty}\bigg[c_5^2\,\sum_{k=1}^{n}\rho^{2k}+2\,c_5\,c_6\,\rho\, \sum_{k=1}^{n}\sum_{l=1}^{k}\rho^{k-l}\alpha[l-1]\rho^k\nonumber\\
&\qquad+c_6^2\,\rho^2\sum_{k=1}^{n}\sum_{l=1}^{k}\sum_{t=1}^{k}\rho^{k-l}\rho^{k-t}\alpha[l-1]\,\alpha[t-1]\bigg].
\end{align}
Since $\rho^2<1$, the first term on the RHS of (\ref{x_ort_norm2_sum}) has a finite limit, and so does the second term, due to Lemma \ref{Lemma_sequences}(b.2) and (\ref{step-size}); the third term   has a finite limit too:
\begin{align}
&\lim_{n\rightarrow\infty}\sum_{k=1}^{n}\sum_{l=1}^{k}\sum_{t=1}^{k}\rho^{k-l}\rho^{k-t}\alpha[l-1]\alpha[t-1] \nonumber\\
&\qquad\quad\overset{(a)}{\leq} \lim_{n\rightarrow\infty}c_7\sum_{k=1}^{n}\sum_{l=1}^{k}\rho^{k-l}\alpha[l-1]^2\overset{(b)}{<}\infty, \nonumber
\end{align}
for some finite  positive constant $c_7$, where (a) follows from $\alpha[l]\alpha[t]\leq (\alpha[l]^2+\alpha[t]^2)/2$ and $\sum_{t=1}^{k}\rho^{k-t}\leq\frac{1}{1-\rho}$; and (b) is a consequence of  Lemma \ref{Lemma_sequences}(b.1) and (\ref{step-size}).

\noindent \textbf{Point (c):} Because of space limitation, we prove only  (\ref{P2c1}); (\ref{P2c2}) can be proved following similar arguments as for (\ref{P2b2}). By the minimum principle applied to (\ref{best_resp_x_hat_2}) (evaluated at $\bar{\bx}[n]$) and (\ref{best_resp_x_hat_3}), it follows that
\begin{align}
&\left(\widehat{\bx}_i(\bar{\bx}[n])-\widetilde{\bx}^{\rm av}_i[n]\right)^T\nonumber\smallskip\\
&\,\,\big(\nabla \widetilde{f}_i(\widehat{\bx}_i(\bar{\bx}[n]);\bar{\bx}[n])-\nabla \widetilde{f}_i(\widetilde{\bx}^{\rm av}_i[n];\bar{\bx}[n])\nonumber\smallskip\\
&\,\,\quad + \partial G(\widehat{\bx}_i(\bar{\bx}[n]))- \partial G(\widetilde{\bx}^{\rm av}_i[n])\big)\nonumber\smallskip\\
& \quad \leq  \left(\widehat{\bx}_i(\bar{\bx}[n])-\widetilde{\bx}^{\rm av}_i[n]\right)^T \left(\widetilde{\boldsymbol{\pi}}_{i}^{\rm av}[n]-\boldsymbol{\pi}_{i}(\bar{\bx}[n])\right).\label{min_princ_avg3}
\end{align}
Invoking the uniform strongly convexity of $\widetilde{f}_i(\bullet;\bar{\bx}[n])$ (with constant $\tau_i$) and the convexity of $G(\bullet)$, (\ref{min_princ_avg3}) yields
\begin{align}\label{bound_psi_pi}
&\|\widehat{\bx}_i(\bar{\bx}[n])-\widetilde{\bx}^{\rm av}_i[n]\|\leq \frac{1}{\tau_i}\|\boldsymbol{\pi}_{i}(\bar{\bx}[n])-\widetilde{\boldsymbol{\pi}}_{i}^{\rm av}[n]\|\nonumber\smallskip\\
& \stackrel{(a)}{\leq} \frac{I}{\tau_i}\|\by^{\rm av}_{i}[n]-\bar{\mathbf{r}}^{\rm av}[n]\|\stackrel{(b)}{\leq}\frac{I}{\tau_i}\|\by^{\rm av}[n]-\mathbf{1}_I\otimes\bar{\mathbf{r}}^{\rm av}[n]\|,
\end{align}
where (a) follows from  (\ref{pi}) and (\ref{aux_averages}); and   (b) is obtained  using the same arguments as in (\ref{min_princ3}).
Proceeding as in  (\ref{y_evol})-(\ref{y_evol3}), $\|\by^{\rm av}[n]-\mathbf{1}_I\otimes\bar{\mathbf{r}}^{\rm av}[n]\|$ can be bounded as
\begin{align}
&\|\by^{\rm av}[n]-\mathbf{1}_I\otimes\bar{\mathbf{r}}^{\rm av}[n]\| \nonumber\\
&\leq c_8\,\rho^n +c_9\sum_{l=1}^{n-1} \rho^{n-l} \|\mathbf{\Delta r}^{\rm av}[l,l-1]\|+c_{9}\,\|\mathbf{\Delta r}^{\rm av}[n,n-1]\|\nonumber\smallskip\\
&\,\,\stackrel{(a)}{\leq} c_8\,\rho^n+c_9\,L^{\max} \sum_{l=1}^{n-1} \rho^{n-l}\, \|\bar{\bx}[l]-\bar{\bx}[l-1]\|\nonumber\\
&\;\qquad +c_{10}\,L^{\max}\, \|\bar{\bx}[n]-\bar{\bx}[n-1]\|\nonumber\smallskip\\
&\,\,\stackrel{(b)}{\leq} c_8\,\rho^n +c_{10}\,L^{\max}\, \alpha[n-1]\left\|\dfrac{1}{I}  \left(\mathbf{1}_I^T\otimes \bI_m\right)\Delta\bx^{\texttt{inx}}[n-1]\right\|\nonumber\smallskip\\
&\quad+c_9\,L^{\max} \sum_{l=1}^{n-1} \rho^{n-l}\, \alpha[l-1]\left\|\dfrac{1}{I}  \left(\mathbf{1}_I^T\otimes \bI_m\right)\Delta\bx^{\texttt{inx}}[l-1]\right\|\nonumber\\
&\stackrel{(c)}{\leq} c_8\rho^n+c_9\,L^{\max}\,c\,\sum_{l=1}^{n-1} \rho^{n-l} \alpha[l-1]+c_{10}\,L^{\max}\,c\,\cdot\alpha[n-1]\nonumber\smallskip\\
&\,\,\stackrel[n\rightarrow\infty]{(d)}{\longrightarrow} 0, \nonumber
\end{align}
for some finite positive constants $c_8$, $c_9$, $c_{10}$, where  in (a) we used  (\ref{aux_averages_vec}) and A3; (b) follows from the   evolution of the average vector $\bar{\bx}[n]$ given by (\ref{x_bar_evo}); in (c) we used Proposition \ref{Prop_cons}(a) [cf. (\ref{P2a})]; and (d) follows from (\ref{step-size}) and Lemma \ref{Lemma_sequences}(a). This completes the proof of (\ref{P2c1}).\\
\noindent \textbf{Point (d):} We first prove (\ref{P2d1}). Proceeding as for (\ref{min_princ_avg3})-(\ref{bound_psi_pi}) and invoking F1 and F3, it holds:
\begin{align}\label{bound_psi_psi}
&\|\widetilde{\bx}_i[n]-\widetilde{\bx}^{\rm av}_i[n]\|\;\leq \dfrac{\widetilde{L}_i}{\tau_i}\|\bx_i[n]-\bar{\bx}[n]\|+\frac{1}{\tau_i}\,\|\widetilde{\boldsymbol{\pi}}_{i}[n]-\widetilde{\boldsymbol{\pi}}_{i}^{\rm av}[n]\|\nonumber\smallskip\\
&\quad \stackrel{(a)}{\leq} c_{10}\,\|\bx_i[n]-\bar{\bx}[n]\|+c_{11}\,\|\by_i[n]-\by_i^{\rm av}[n]\|\nonumber\medskip\\
&\quad \leq c_{10}\,\|\bx[n]-\mathbf{1}_I\otimes\bar{\bx}[n]\|+c_{11}\,\|\by[n]-\by^{\rm av}[n]\|\nonumber\medskip\\
&\quad  \stackrel{(b)}{\leq}  c_{10}\,\|\bx[n]-\mathbf{1}_I\otimes\bar{\bx}[n]\|+c_{11}\,\|\mathbf{1}_I\otimes \left(\bar{\br}[n]-\bar{\br}^{\rm av}[n]\right)\|\nonumber\medskip\\
&\quad \,\,\,+c_{11}\,\|\by[n]-\by^{\rm av}[n]-\mathbf{1}_I\otimes \left(\bar{\br}[n]-\bar{\br}^{\rm av}[n]\right)\|
\end{align}
where $\widetilde{L}_i$ is the Lipschitz constant of $\nabla \widetilde{f}_i(\widetilde{\bx}^{\rm av}_i[n]; \bullet)$ (cf. F3);  in (a) we set  $c_{10}\triangleq \max_i\;(L_i+\widetilde{L}_i)/\tau_i$ and $c_{11}\triangleq \max_i\;I/\tau_i$, with $L_{i}$ being the Lipschitz constant of $\nabla f_i(\bullet)$, and we used
\vspace{-.15cm}
\begin{equation}\nonumber
\|\widetilde{\boldsymbol{\pi}}_{i}[n]-\widetilde{\boldsymbol{\pi}}_{i}^{\rm av}[n]\|\;\leq\; L_{i}\|\bx_i[n]-\bar{\bx}[n]\|+I\|\by_i[n]-\by_i^{\rm av}[n]\|;\vspace{-.15cm}
\end{equation}
and (b) follows from  $\|\ba\|\leq\|\ba-\bb\|+\|\bb\|$, with $\ba=\by[n]-\by^{\rm av}[n]$ and $\bb=\mathbf{1}_I\otimes\left(\bar{\br}[n]-\bar{\br}^{\rm av}[n]\right)$. To complete the proof, it is sufficient to show that the RHS of \eqref{bound_psi_psi} is asymptotically vanishing, which is proved next.  It follows from Proposition \ref{Prop_cons}(b) [cf. (\ref{P2b1})] that $\|\bx[n]-\mathbf{1}_I\otimes\bar{\bx}[n]\| {\longrightarrow} 0$ as $n\rightarrow\infty$. The second term on the RHS of \eqref{bound_psi_psi} can be bounded as \vspace{-.2cm}
\begin{align}\label{RHS_sec_term}
&\|\mathbf{1}_I\otimes \left(\bar{\br}[n]-\bar{\br}^{\rm av}[n]\right)\| \leq \sum_{i=1}^I  \| \nabla f_i[n]-\nabla f_i^{\rm av}[n]  \|\nonumber \vspace{-.3cm}\\
& \leq \sum_{i=1}^I L_{i}\|\bx_i[n]-\bar{\bx}_i[n]\|  \underset{n\rightarrow\infty}{\longrightarrow} 0, \vspace{-.2cm}
\end{align}
and the last implication is due to Proposition \ref{Prop_cons}(b) [cf. (\ref{P2b1})]. We bound now the last term on the RHS of    \eqref{bound_psi_psi}.
Using (\ref{aux_averages}), (\ref{aux_averages_vec}) and (\ref{y_evol}),  it is not difficult to show that
\begin{align}
&\by[n]-\by^{\rm av}[n]=\widehat{\mathbf{P}}[n-1,0] \left(\br[0]-\br^{\rm av}[0]\right)\smallskip\nonumber\\
&\qquad +\sum_{l=1}^{n-1} \widehat{\mathbf{P}}[n-1,l] \left(\Delta \br[l,l-1]-\Delta\br^{\rm av}[l,l-1]\right)\smallskip\nonumber\\
&\qquad +\left(\Delta \br[n,n-1]-\Delta\br^{\rm av}[n,n-1]\right),\nonumber
\end{align}
which, exploiting the same arguments used in (\ref{Kron_r_bar})-(\ref{y_evol3}),  leads to the following bound:
\begin{align}\label{RHS_third}
&\|\by[n]-\by^{\rm av}[n]-\mathbf{1}_I\otimes \left(\bar{\br}[n]-\bar{\br}^{\rm av}[n]\right)\|\nonumber\smallskip\\
& \leq c_{12}\,\rho^n +c_{13}\,\sum_{l=1}^{n-1} \rho^{n-l}\, \left\Vert \Delta \br[l,l-1]-\Delta\br^{\rm av}[l,l-1]\right\Vert\nonumber\\
&\quad +c_{14}\,\left\Vert\Delta \br[n,n-1]-\Delta\br^{\rm av}[n,n-1]\right\Vert,\nonumber\\
& \stackrel{(a)}{\leq}  c_{12}\,\rho^n \nonumber\smallskip\\&+c_{15}\,\sum_{i=1}^{I}\sum_{l=1}^{n-1} \rho^{n-l}\, \left(\left\Vert  \bx_i[l]-\bar{\bx}[l]\right\Vert + \left\Vert  \bx_i[l-1]-\bar{\bx}[l-1]\right\Vert\right)\nonumber\\
&\quad +c_{16}\,\sum_{i=1}^{I} \left(\left\Vert  \bx_i[n]-\bar{\bx}[n]\right\Vert + \left\Vert  \bx_i[n-1]-\bar{\bx}[n-1]\right\Vert\right)\nonumber\medskip\\
&\quad\stackrel[n\rightarrow\infty]{(b)}{\longrightarrow} 0,
\end{align}
for some finite positive constants  $c_{12}$, $c_{13}$, $c_{14}$, $c_{15}$, and $c_{16}$; where (a) follows from the definitions of $\Delta \br[l,l-1]$ and  $\Delta\br^{\rm av}[l,l-1]$ [cf. (\ref{delta_r}) and (\ref{aux_averages_vec})], and the Lipschitz property of $\nabla f_i(\bullet)$ (cf. A3); and (b) follows from  Proposition \ref{Prop_cons}(b) [cf. (\ref{P2b1})] and Lemma   \ref{Lemma_sequences}(a).

We prove now (\ref{P2d2}). From (\ref{bound_psi_psi}), (\ref{RHS_sec_term}), and (\ref{RHS_third}),  we get
\begin{align}\label{bound_psi_psi4}
&\lim_{n\rightarrow\infty}\sum_{k=1}^{n}\alpha[k]\|\widetilde{\bx}_i[k]-\widetilde{\bx}^{\rm av}_i[k]\| \leq \lim_{n\rightarrow\infty}c_{17}\bigg[\sum_{k=1}^{n}\rho^k\alpha[k]\nonumber\\
&\hspace{-.3cm}+\sum_{i=1}^I\sum_{k=1}^{n}\sum_{l=1}^{k-1} \rho^{k-l}\alpha[k]\Big(\|\bx_i[l]-\bar{\bx}[l]\|+\|\bx_i[l-1]-\bar{\bx}[l-1]\|\Big) \nonumber\\
&+\sum_{i=1}^I\sum_{k=1}^{n} \alpha[k]\Big(\|\bx_i[k]-\bar{\bx}[k]\|+ \|\bx_i[k-1]-\bar{\bx}[k-1]\|\Big)\bigg]\nonumber\medskip\\
&\;\; < \infty,
\end{align}
for some constant $c_{17}=\max(c_{12},c_{15},c_{16})$, where the last implication follows from the finiteness of the three terms on the RHS of (\ref{bound_psi_psi4}). Indeed, under  (\ref{step-size}) and (\ref{P2b2}), the first and third term have a finite limit, whereas lemma \ref{Lemma_sequences}(b.2) along with  property (\ref{P2b3}) guarantee the finiteness of  the second term. This completes the proof of  (\ref{P2d2}).
\end{proof}

\vspace{-0.35cm}
\subsection{Proof of Theorems \ref{simplified_convergence_th} and \ref{main_th}}

We prove only Theorem \ref{main_th}; the proof of Theorem \ref{simplified_convergence_th} comes as special case. Invoking the descent lemma on $F$ and {using (\ref{x_y_bar})}, (\ref{x_bar_evo}), we get
\begin{align}
& F(\bar{\bx}[n+1])\leq F(\bar{\bx}[n]) \nonumber\\
&\quad+\dfrac{\alpha[n]}{I}\,\nabla F(\bar{\bx}[n])^T\,\sum_{i=1}^I\left(\bx_i^{\texttt{inx}}[n]-\bar{\bx}[n]\right) \nonumber\\
&\quad+ \frac{L_{\max}}{2}\left(\dfrac{\alpha[n]}{I}\right)^2\sum_{i=1}^I\left\|\bx_i^{\texttt{inx}}[n]-\bx_i[n]\right\|^2.\nonumber
\end{align}
Summing and subtracting the quantity $\dfrac{\alpha[n]}{I}\nabla F(\bar{\bx}[n])^T\sum_{i=1}^I \Big(\widetilde{\bx}_i[n] +\widetilde{\bx}^{\rm av}_i[n]+ \widehat{\bx}_i(\bar{\bx}[n])\Big)$ to the RHS of the above inequality, we obtain
\begin{align}\label{desc_lemma4}
&F(\bar{\bx}[n+1])\leq F(\bar{\bx}[n]) \nonumber\\
&\quad+\dfrac{\alpha[n]}{I}\nabla F(\bar{\bx}[n])^T\sum_{i=1}^I\left(\widehat{\bx}_i(\bar{\bx}[n])-\bar{\bx}[n]\right)\nonumber\\
&\quad+ \dfrac{\alpha[n]}{I}\,\nabla F(\bar{\bx}[n])^T\,\sum_{i=1}^I\left(\widetilde{\bx}^{\rm av}_i[n]-\widehat{\bx}_i(\bar{\bx}[n]) \right)\nonumber\\
&\quad+\dfrac{\alpha[n]}{I}\nabla F(\bar{\bx}[n])^T\,\sum_{i=1}^I\,\left(  \widetilde{\bx}_i[n] - \widetilde{\bx}^{\rm av}_i[n] \right)\nonumber\\
&\quad+\dfrac{\alpha[n]}{I}\nabla F(\bar{\bx}[n])^T\,\sum_{i=1}^I\,\left(\bx_i^{\texttt{inx}}[n]- \widetilde{\bx}_i[n]\right)\nonumber\\
&\quad +\frac{L_{\max}}{2}\left(\dfrac{\alpha[n]}{I}\right)^2\sum_{i=1}^I\left\|\bx_i^{\texttt{inx}}[n]-\bx_i[n]\right\|^2
\end{align}

The second term on the RHS of (\ref{desc_lemma4}) can be bounded as {
\begin{align}\label{first_ord_term}
&\dfrac{\alpha[n]}{I}\,\nabla F(\bar{\bx}[n])^T\,\sum_{i=1}^I\left(\widehat{\bx}_i(\bar{\bx}[n])-\bar{\bx}[n] \right)\nonumber\\
& \quad \qquad\stackrel{(a)}{\leq} - \frac{c_{\tau}}{I}\, \alpha[n]\,\sum_{i=1}^I\left\|\widehat{\bx}_i(\bar{\bx}[n])-\bar{\bx}[n] \right\|^2\nonumber\\
& \qquad \qquad   \qquad + \alpha[n] \left( G(\bar{\bx}[n])- \dfrac{1}{I} \sum_{i=1}^I G(\widehat{\bx}_i(\bar{\bx}[n]))\right)\nonumber\medskip\\
& \quad \qquad\stackrel{(b)}{\leq}- \frac{c_{\tau}}{I}\, \alpha[n]\,\sum_{i=1}^I\left\|\widehat{\bx}_i(\bar{\bx}[n])-\bar{\bx}[n] \right\|^2\nonumber\\
& \quad \qquad   \qquad +   G(\bar{\bx}[n])-   G(\bar{\bx}[n+1])\nonumber\\
& \quad \qquad   \qquad +  \frac{1}{I}\, \alpha[n] \sum_{i=1}^I  \left(G(\bx_i^{\texttt{inx}}[n]) - G(\widehat{\bx}_i(\bar{\bx}[n]) ) \right)
\nonumber\\
%\end{align}
%\begin{align}
& \quad \qquad\stackrel{(c)}{\leq}- \frac{c_{\tau}}{I}\, \alpha[n]\,\sum_{i=1}^I\left\|\widehat{\bx}_i(\bar{\bx}[n])-\bar{\bx}[n] \right\|^2\nonumber\\
& \quad \qquad   \qquad +   G(\bar{\bx}[n])-   G(\bar{\bx}[n+1])\nonumber\\
& \quad \qquad   \qquad +  \frac{L_G}{I}\, \alpha[n] \sum_{i=1}^I  \left\|\bx_i^{\texttt{inx}}[n] - \widehat{\bx}_i(\bar{\bx}[n])  \right\|,
\end{align}
where (a) follows from Proposition \ref{Prop_best_response} [cf. (\ref{p1c})]; in (b) we used the following inequality, due to the convexity of $G$ and (\ref{x_bar_evo}):
\begin{equation}\nonumber
G(\bar{\bx}[n+1])\leq (1-\alpha[n])\,G(\bar{\bx}[n])+\dfrac{\alpha[n]}{I}\,\sum_{i=1}^I G\left(\bx_i^{\texttt{inx}}[n]\right),
\end{equation}
and (c) follows from the convexity of $G$ and $\|\partial G(\bx_i^{\texttt{inx}}[n])\|\leq L_G$ (cf. A5).

Exploiting (\ref{first_ord_term}) in (\ref{desc_lemma4}), and using A4, Proposition \ref{Prop_cons}(a) [cf. (\ref{P2a})], the triangle inequalities, and (\ref{precision}), we can write
\begin{align}\label{desc_lemma5}
&U(\bar{\bx}[n+1])\leq  U(\bar{\bx}[n]) \nonumber\\
&\quad-\frac{c_{\tau}}{I}\, \alpha[n] \sum_{i=1}^I\|\widehat{\bx}_i(\bar{\bx}[n])-\bar{\bx}[n]\|^2 \nonumber\\
&\quad +c_{18}\,\alpha[n] \sum_{i=1}^I\|\widetilde{\bx}^{\rm av}_i[n]-\widehat{\bx}_i(\bar{\bx}[n])\| +c_{18}\,\alpha[n]\sum_{i=1}^I \varepsilon_i[n]   \nonumber\\
&\quad +c_{18}\,\alpha[n]\sum_{i=1}^I\|\widetilde{\bx}_i[n]- \widetilde{\bx}^{\rm av}_i[n]\| + \frac{c^2L_{\max}}{2I}\alpha^2[n],
\end{align}
with $c_{18}\triangleq \max(L_F,L_G)/I$. We apply now Lemma \ref{lemma_Robbinson_Siegmunt} with the following identifications:
\begin{align}\nonumber
&Y[n]=U(\bar{\bx}[n]) \nonumber\\
&X[n]=c_{18}\,c_{\tau}\,\dfrac{\alpha[n]}{I}\sum_{i=1}^I\|\widehat{\bx}_i(\bar{\bx}[n])-\bar{\bx}[n]\|^2 \nonumber\\
&Z[n]=\frac{c^2L_{\max}}{2I}\alpha[n]^2+c_{18}\,\alpha[n]\sum_{i=1}^I\|\widetilde{\bx}^{\rm av}_i[n]-\widehat{\bx}_i(\bar{\bx}[n])\|  \nonumber\\
&\quad +c_{18}\,\alpha[n]\sum_{i=1}^I\|\widetilde{\bx}_i[n]- \widetilde{\bx}^{\rm av}_i[n]\|+c_{18}\,\alpha[n]\sum_{i=1}^I \varepsilon_i[n]. \nonumber
\end{align}
Note that  $\sum_{n=1}^\infty Z[n]<\infty$, due to (\ref{step-size}), (\ref{errors}), and properties (\ref{P2c2}) and (\ref{P2d2}). Since   $U(\bar{\bx}[n])$ is coercive (cf. A6), it follows from Lemma \ref{lemma_Robbinson_Siegmunt} that $U(\bar{\bx}[n])$  converges to a finite value, and
\begin{align}\nonumber
\sum_{n=1}^\infty \alpha[n]\, \|\widehat{\bx}_i(\bar{\bx}[n])-\bar{\bx}[n]\|^2 < \infty, \quad \forall i=1,\ldots,I,
\end{align}
which, using (\ref{step-size}), leads to
\begin{align}\label{conv2}
\underset{n\rightarrow\infty}{{\mbox{lim inf}}}\;\,\|\widehat{\bx}_i(\bar{\bx}[n])-\bar{\bx}[n]\| =0, \quad  \forall i=1,\ldots,I.
\end{align}
Following similar  arguments as in \cite{Scutari-Facchinei-Song-Palomar-Pang,Scutari-Facchinei-Sagratella}, it is not difficult to show that $\underset{n\rightarrow\infty}{{\mbox{lim sup}}}\;\,\|\widehat{\bx}_i(\bar{\bx}[n])-\bar{\bx}[n]\| =0$, for all $i$. Thus, we have
$\underset{n\rightarrow\infty}{{\mbox{lim}}}\;\,\|\widehat{\bx}_i(\bar{\bx}[n])-\bar{\bx}[n]\| =0,$
for all  $i$. Since the sequence $\{\bar{\bx}[n]\}_n$ is bounded (due to  A6 and   the convergence of $\{U(\bar{\bx}[n])\}_n$), it has at
least one limit point $\bar{\bx}^\infty\in \mathcal{K}$. By continuity of each $\widehat{\bx}_i(\bullet)$ [cf. Proposition \ref{Prop_best_response}(a)]   and $\underset{n\rightarrow\infty}{{\mbox{lim}}}\;\,\|\widehat{\bx}_i(\bar{\bx}[n])-\bar{\bx}[n]\| =0$, it must be $\widehat{\bx}_i(\bar{\bx}^\infty)=\bar{\bx}^\infty$ for all $i$. By Proposition \ref{Prop_best_response}(b), $\bar{\bx}^\infty$ is also a stationary solution of Problem (1). This proves statement (a) of the theorem. Statement (b) follows readily from (\ref{P2b1}).

\balance

\bibliographystyle{MyIEEE}
\bibliography{biblio}

\vspace{-.6cm}
\begin{IEEEbiography}[{\includegraphics[scale=0.078]{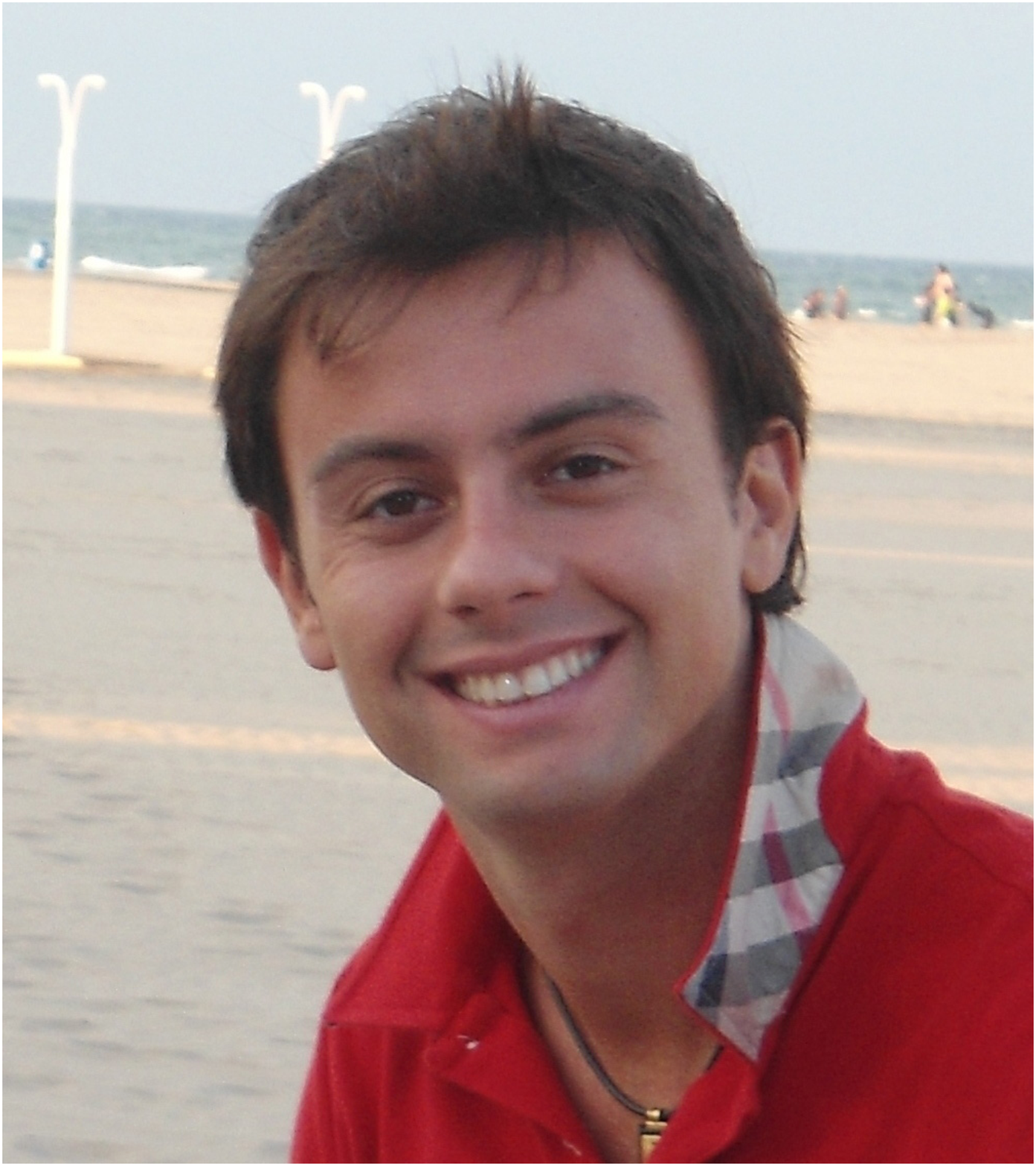}}]{Paolo Di Lorenzo}
(S'10-M'13)  received the M.Sc. degree in 2008 and the Ph.D. in electrical engineering in 2012, both from University of Rome ``La Sapienza,'' Italy. He is an assistant Professor in the Department of Engineering, University of Perugia, Italy. During 2010 he held a visiting research appointment in the Department of Electrical Engineering, University of California at Los Angeles (UCLA). He has participated in the European research project FREEDOM, on femtocell networks, SIMTISYS, on moving target detection through satellite constellations, and TROPIC, on distributed computing, storage and radio resource allocation over cooperative femtocells. His primary research interests are in distributed optimization and learning over networks, statistical signal processing, graph theory, game theory, and adaptive filtering. Dr. Di Lorenzo received three best student paper awards, respectively at IEEE SPAWC'10, EURASIP EUSIPCO'11, and IEEE CAMSAP'11, for works in the area of signal processing for communications and synthetic aperture radar systems. He is recipient of the 2012 GTTI (Italian national group on telecommunications and information theory) award for the Best Ph.D. Thesis in information technologies and communications.
\end{IEEEbiography}

\vspace{-.6cm}
\begin{IEEEbiography}[{\includegraphics[scale=0.155]{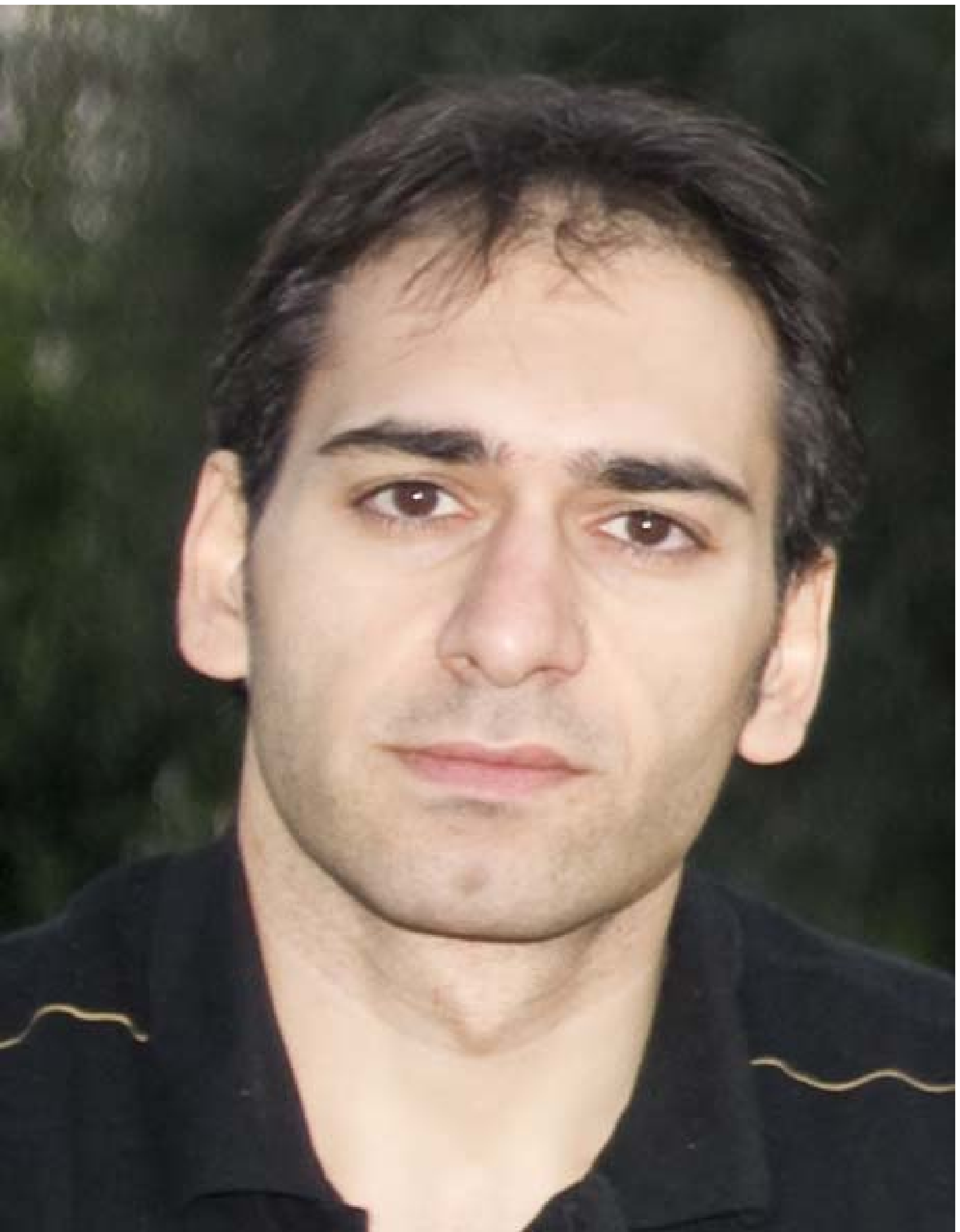}}]{Gesualdo Scutari}
(S'05-M'06-SM'11) received the Electrical Engineering and Ph.D. degrees (both with honors) from the University of Rome ``La Sapienza,'' Rome, Italy, in 2001 and 2005, respectively. He is an Associate Professor with the Department of Industrial Engineering, Purdue University, West Lafayette, IN, USA; and he is the scientific director  for the area of Big-Data Analytics at the Cyber Center (Discovery Park) at Purdue University. He had previously held several research appointments, namely, at the University of California at Berkeley, Berkeley, CA, USA; Hong Kong University of Science and Technology, Hong Kong; University of Rome, ``La Sapienza,'' Rome, Italy; University of Illinois at Urbana-Champaign, Urbana, IL, USA.  His research interests include theoretical and algorithmic issues related to big data optimization, equilibrium programming, and their applications to signal processing, medical imaging, machine learning, and networking. Dr. Scutari is an Associate Editor of the IEEE TRANSACTIONS ON SIGNAL PROCESSING and he served as an Associate Editor of the IEEE SIGNAL PROCESSING LETTERS. He served on the IEEE Signal Processing Society Technical Committee on Signal Processing for Communications (SPCOM). He was the recipient of the 2006 Best Student Paper Award at the International Conference on Acoustics, Speech, and Signal Processing (ICASSP) 2006, the 2013 NSF Faculty Early Career Development (CAREER) Award, the 2013 UB Young Investigator Award, the 2015  AnnaMaria Molteni Award for Mathematics and Physics (from ISSNAF), and the 2015 IEEE Signal Processing Society Young Author Best Paper Award.
\end{IEEEbiography}

\end{document}